\documentclass[12pt]{article}
\usepackage{fullpage, amsmath, amsfonts, amssymb, amsthm, natbib}
\usepackage{graphicx, enumerate, color, url, multirow, lscape,rotating, float}
\usepackage[lined,boxed,linesnumbered]{algorithm2e}

\newcommand{\real}{\mathbb R}

\newcommand{\E}{\mathbb{E}}

\newcommand{\V}{\mathrm{Var}}
\newcommand{\C}{\mathrm{Cov}}

\newcommand{\ve}{\mathrm{vec}}

\newcommand{\iGCV}{\textnormal{iGCV}}

\newtheorem{lem}{Lemma}

\newtheorem{prop}{Proposition}

\def\bB{\mathbf{B}}
\def\bb{\mathbf{b}}
\def\bC{\mathbf{C}}

\def\bD{\mathbf{D}}

\def\bF{\mathbf{F}}

\def\bL{\mathbf{L}}
\def\bbf{\mathbf{f}}

\def\bI{\mathbf{I}}
\def\bP{\mathbf{P}}
\def\bQ{\mathbf{Q}}

\def\bU{\mathbf{U}}
\def\bu{\mathbf{u}}
\def\bV{\mathbf{V}}
\def\bS{\mathbf{S}}
\def\bG{\mathbf{G}}
\def\bg{\mathbf{g}}
\def\bX{\mathbf{X}}
\def\bx{\mathbf{x}}

\def\by{\mathbf{y}}

\def\bs{\mathbf{s}}

\def\bw{\mathbf{w}}

\def\bmu{\boldsymbol{\mu}}
\def\balpha{\boldsymbol{\alpha}}
\def\bbeta{\boldsymbol{\beta}}

\def\bGamma{\boldsymbol{\Gamma}}

\def\bSigma{\boldsymbol{\Sigma}}
\def\bTheta{\boldsymbol{\Theta}}
\def\btheta{\boldsymbol{\theta}}
\def\bPhi{\boldsymbol{\Phi}}
\def\bPsi{\boldsymbol{\Psi}}
\def\bLambda{\boldsymbol{\Lambda}}

\def\bxi{\boldsymbol{\xi}}
\def\bfeta{\boldsymbol{\eta}}
\def\ba{\boldsymbol{a}}

\def\bChat{\widehat{\mathbf{C}}}

\def\bBtilde{\widetilde{\mathbf{B}}}

\def\bftilde{\widetilde{\mathbf{f}}}
\def\bCtilde{\widetilde{\mathbf{C}}}

\def\bLtilde{\widetilde{\mathbf{L}}}
\def\bdtilde{\widetilde{\mathbf{d}}}

\def\bThetatilde{\widetilde{\mathbf{\Theta}}}
\def\bbtilde{\widetilde{\bb}}

\def\bPhitilde{\widetilde{\boldsymbol{\Phi}}}
\def\bPtilde{\widetilde{\mathbf{P}}}

\def\bbftilde{\widetilde{\mathbf{f}}}
\def\brho{\boldsymbol{\rho}}
\def\T{\mathcal{T}}

\def\H{\mathcal{H}}

\def\bZ{\mathbf{Z}}

\title{Fast Covariance Estimation for Multivariate Sparse Functional Data}
\author{Cai Li$^{1,*}$,
	Luo Xiao$^{1}$ and Sheng Luo$^{2}$\\
	$^{1}$ North Carolina State University and $^{2}$Duke University\\
	$^{*}$email: cli9$@$ncsu.edu
}
\begin{document}
\maketitle

\abstract
Covariance estimation is essential yet underdeveloped for analyzing multivariate functional data.
We propose a fast covariance estimation method for multivariate sparse functional data using bivariate penalized splines.
The tensor-product B-spline formulation of the proposed method
enables a simple spectral decomposition of the associated covariance operator 
and explicit expressions of the resulting eigenfunctions as linear combinations of B-spline bases,
thereby dramatically facilitating subsequent principal component analysis.
We derive a fast algorithm for selecting the smoothing parameters in covariance smoothing
 using leave-one-subject-out cross-validation. 
The method is evaluated with extensive numerical studies and applied to an Alzheimer's disease study with multiple longitudinal outcomes.


\section{Introduction}
Functional data analysis (FDA) has been enjoying great successes in many applied fields, e.g.,
neuroimaging \citep{Reiss:10, Lindquist:12, Goldsmith:12, zhu2012multivariate},  genetics \citep{Leng:06, Reimherr:14, Reimherr:16}, and wearable computing \citep{Morris:06, Xiao:15}. 
Functional principal component analysis (FPCA) conducts dimension reduction on the inherently infinite-dimensional functional data,
and thus facilitates subsequent modeling and analysis.
Traditionally, functional data are densely observed on a common grid and can be easily connected to multivariate data, 
although the notion of smoothness distinguishes the former from the latter.
In recent years,  covariance-based FPCA \citep{Yao:05a} has become a standard approach
and has greatly expanded the applicability of functional data methods to irregularly spaced data such as longitudinal data.
Various nonparametric methods have now been proposed to estimate the smooth covariance function, e.g., 
\cite{Peng:09}, \cite{cai2010nonparametric}, \cite{Goldsmith:12}, \cite{Xiao:16b}
and \cite{Wong2019}.

There has been  growing interest in multivariate functional data  
where multiple functions are observed for each subject. 
For dense functional data,
\citet[Chapter 8.5]{Ramsay:05}
proposed to concatenate multivariate functional data as a single vector and conduct multivariate PCA on the long vectors
and \cite{Berrendero:11} repeatedly applied point-wise univariate PCA. 
For sparse and paired functional data, \cite{zhou2008joint} extended the low-rank mixed effects model in \cite{James:00}.
\cite{Chiou:14} considered normalized multivariate FPCA through standardizing the covariance operator. 
\cite{petersen2016frechet} proposed various metrics for studying cross-covariance between multivariate functional data.
More recently, \cite{Happ:17} introduced a FPCA framework for multivariate functional data defined on different domains.

The interest of the paper is  functional principal component analysis for multivariate sparse functional data, 
where multiple responses are observed at time points that vary from subjects to subjects and may even vary between responses within subjects.
There are much fewer works to handle such data. 
The approach in \cite{zhou2008joint} focuses on bivariate functional data and
can be extended to more than two-dimensional functional data, although model selection (e.g., selection of smoothing parameters) can
be computationally difficult and convergence of the expectation-maximization estimation algorithm could also be an issue.
The local polynomial method in \cite{Chiou:14} can be applied to multivariate sparse functional data,
although a major drawback is the selection of multiple bandwidths. Moreover, because
the local polynomial method is a local approach, there is no guarantee that the resulting estimates of covariance
functions will lead to a properly defined covariance operator. 
The approach in \cite{Happ:17} (denoted by mFPCA hereafter) estimates cross-covariances via scores from univariate FPCA and hence can be
applied to multivariate sparse functional data. 
While mFPCA is theoretically sound for dense functional data,
it may not capture cross-correlations between functions
because scores from univariate FPCA for sparse functional data 
are shrunk towards zero.

We propose a novel and fast covariance-based FPCA method for multivariate sparse functional data. 
Note that multiple auto-covariance functions for within-function correlations and
cross-covariance functions for between-function correlations have to be estimated.
Tensor-product B-splines are employed to approximate the covariance functions
and a smoothness penalty as in bivariate penalized splines \citep{Eilers:03} is adopted
to avoid overfit. 
Then the individual estimates of covariance functions will be pooled and refined.
The advantages of the new method are multifold. 
First, the tensor-product B-spline formulation is computationally efficient to handle multivariate sparse functional data.
Second, a fast fitting algorithm for selecting the smoothing parameters will be derived, 
which alleviates the computational burden of conducting leave-one-subject-out cross-validation.
Third, the tensor-product B-spline representation of the covariance functions enables a straightforward spectral decomposition 
of the covariance operator for the multivariate functional data; see Proposition \ref{eq:eigen}. In particular, the eigenfunctions associated with
the covariance operator are explicit functions of the B-spline bases.
Last but not the least, via a simple truncation step, the refined estimates of the covariance functions
lead to a properly defined covariance operator.

Compared to mFPCA, the proposed method does not rely on  scores from univariate FPCA, which could be
a severe problem for sparse functional data, and hence
could better capture the correlations between functions. 
And an improved correlation estimation
will lead to improved subsequent FPCA analysis and curve prediction.
The proposed method also compares favorably with the local polynomial method in \cite{Chiou:14} because of
the computationally efficient tensor-product spline formulation of the covariance functions and
the derived fast algorithm for selecting the smoothing parameters.
Moreover, as mentioned above, 
there exists  an explicit and easy-to-calculate relationship between
the tensor-product spline representation of covariance functions
and the associated eigenfunctions/eigenvalues, which greatly facilitates 
subsequent FPCA analysis.

In addition to FPCA, there are also abundant literatures on models for multivariate functional data with most focusing on dense functional data.
For clustering of multivariate functional data, see \cite{Zhu:12, jacques2014model,huang2014joint} and \cite{park2017clustering}.
For regression with multivariate functional responses, see \cite{zhu2012multivariate, luo2017function, li2017functional, wong2017partially, zhu2017multivariate, kowal2017bayesian} and \cite{qi2018function}.
Graphical models for multivariate functional data are studied in
\cite{zhu2016bayesian} and \cite{qiao2017functional}.
Works on multivariate functional data include also \cite{chiou2014, chiou2016pairwise}.

The remainder of the paper proceeds as follows. In Section \ref{sec:method}, we present our proposed method.
We conduct extensive simulation studies in Section \ref{sec:simulation}
and apply the proposed method to an Alzheimer's disease study in Section \ref{sec:application}. 
A discussion is given in Section \ref{sec:discussion}. 
All  technical details are enclosed in the Appendix.

\section{Methods} \label{sec:method}
\subsection{Fundamentals of Multivariate Functional Principal Component Analysis}
Let $p$ be a positive integer and denote by $\T$ a continuous and bounded domain in the real line $\real$.
Consider the Hilbert space $\H:\underbrace{L^2(\T) \times \ldots \times L^2(\T)}_{p}$
equipped with the inner product $<\cdot,\cdot>_{\H}$ and norm $\|\cdot\|_{\H}$
such that for arbitrary functions $\bbf = \left(f^{(1)},\ldots,f^{(p)}\right)^{\top}$ and $\bg = \left(g^{(1)},\ldots,g^{(p)}\right)^{\top}$ in $\H$ with each element in $L^2(\T)$, 
$<\bbf, \bg>_{\H} = \sum_{k=1}^p \int f^{(k)}(t) g^{(k)}(t) dt$ and  $\|\bbf\|_{\H} = <\bbf,\bbf>_{\H}^{1/2}$.
Let $\left\{x^{(k)}\right\}_{k = 1,\ldots,p}$ be a set of $p$ random functions with each function in $L^2(\T)$.
Assume that the $p$-dimensional vector $\bx(t) = \left(x^{(1)},\ldots, x^{(p)}\right)^{\top} \in \real^p$ 
has a $p$-dimensional smooth mean function, $\bmu(t) = \E\{\bx(t)\} = \left(\E\left\{x^{(1)}(t)\right\}, \ldots, \E\left\{x^{(p)}(t)\right\}\right)^{\top} = \left(\mu^{(1)}(t), \ldots, \mu^{(p)}(t)\right)^{\top}$. 
Define the covariance function as $\bC(s,t) = \E\left\{ (\bx(s) - \bmu(s)) (\bx(t) - \bmu(t))^{\top} \right\} = \left[ C_{kk^{\prime}}(s,t)\right]_{1\leq k,k^{\prime}\leq p}$ and $C_{kk^{\prime}}(s, t) = \C\left\{x^{(k)}(s),x^{(k^{\prime})}(t)\right\}$.
Then the covariance operator $\bGamma: \H \rightarrow \H$ associated with the kernel $\bC(s,t)$ can be defined
such that for any $\bbf \in \H$, the $k$th element of $\bGamma \bbf$ is given by
\begin{eqnarray*}
	(\bGamma\bbf)^{(k)}(s) = <\bC_k(s,\cdot), \bbf>_{\H}  = \sum_{k^{\prime}=1}^p \int C_{kk^{\prime}}(s,t) f^{(k^{\prime})}(t) dt,
\end{eqnarray*}
where $\bC_k(s,t) = (C_{k1}(s,t),\ldots,C_{kp}(s,t))^{\top}$. 
Note that $\bGamma$ is a linear, self-adjoint, compact and non-negative integral operator.
By the Hilbert-Schmidt theorem, there exists a set of orthonormal bases $\{\bPsi_{\ell}\}_{\ell \geq 1} \in \H$, $\bPsi_{\ell} = \left(\Psi_{\ell}^{(1)},\ldots, \Psi_{\ell}^{(p)}\right)^{\top}$, and $<\bPsi_{\ell}, \bPsi_{\ell^{\prime}}>_{\H} = \sum_{k=1}^p\int \Psi_{\ell}^{(k)}(t)\Psi_{\ell^{\prime}}^{(k)}(t) dt = {1}_{\left\{\ell=\ell^{\prime}\right\}}$, such that 
\begin{eqnarray} \label{eq:cov_operator}
	(\bGamma\bPsi_{\ell})^{(k)}(s) =  \sum_{k^{\prime}=1}^p \int C_{kk^{\prime}}(s,t) \Psi_{\ell}^{(k^{\prime})}(t) dt = d_{\ell} \Psi_{\ell}^{(k)}(s),
\end{eqnarray}
where $d_{\ell}$ is the $\ell$th largest eigenvalue corresponding to $\bPsi_{\ell}$.
Then the multivariate Mercer's theorem gives
\begin{eqnarray} \label{eq:mult_cov}
	\bC(s,t) = \sum_{\ell}^{\infty}d_{\ell}\bPsi_{\ell}(s)\bPsi_{\ell}^{\top}(t),
\end{eqnarray}
where $C_{k k^{\prime}}(s,t) = \sum_{\ell=1}^{\infty}d_{\ell}\Psi_{\ell}^{(k)}(s)\Psi_{\ell}^{(k^{\prime})}(t)$.
As shown in \cite{Saporta:81}, $\bx(t)$ has the multivariate Karhunen-Lo\`eve  representation, $\bx(t) = \bmu(t) + \sum_{\ell=1}^{\infty}\xi_{\ell}\bPsi_{\ell}(t)$,  where  $\xi_{\ell} = <\bx - \bmu, \bPsi_{\ell}>_{\H}$ are the  scores  with $\E(\xi_{\ell}) = 0$ and $\E(\xi_{\ell}\xi_{\ell^{\prime}}) = d_{\ell}{1}_{\left\{\ell=\ell^{\prime}\right\}}$.
The covariance operator $\bGamma$ has the positive semi-definiteness property, i.e,
for any $\ba = (a_1,\ldots, a_p)^{\top}\in\real^p$, the covariance function of $\ba^{\top}\bx$,
denoted by $C_{\ba}(s,t)$,
satisfies that for any sets of time points $(t_1,\ldots, t_q)\subset \T$ with 
an arbitrary positive integer $q$, the square matrix $[C_{\ba}(t_i, t_j)]_{\{1\leq i, j\leq q\}}\in \real^{q\times q}$ is positive semi-definite.

\subsection{Covariance Estimation by Bivariate Penalized Splines}
Suppose that the observed data take the form $\left\{\left(y_{ij}^{(k)},t_{ij}^{(k)}\right): i = 1, \ldots, n;\,\, k = 1, \ldots, p;\,\, j=1,\ldots, m_{ik}\right\}$,  where $t_{ij}^{(k)} \in \T$ is the observed time point,
$y_{ij}^{(k)}$ is the observed $k$th  response,
$n$ is the number of subjects, and $m_{ik}$ is the number of observations for subject $i$'s $k$th response. 
The model is 
\begin{eqnarray}\label{eq:model}
	y_{ij}^{(k)}= x_i^{(k)}\left(t_{ij}^{(k)}\right)+ \epsilon_{ij}^{(k)}= \mu^{(k)}\left(t_{ij}^{(k)}\right) + \sum_{\ell=1}^{\infty}\xi_{i \ell}\Psi_{\ell}^{(k)}\left(t_{ij}^{(k)}\right) + \epsilon_{ij}^{(k)},
\end{eqnarray}
where $\bx_{i}(t) = \left(x_i^{(1)}(t), \ldots, x_i^{(p)}(t)\right)^{\top}\in \H$,   $\epsilon_{ij}^{(k)}$ are random noises with zero means and variances $\sigma_{k}^2$ and  are  independent across $i, j$ and $k$. 

The goal is to estimate the  covariance functions $C_{k k'}$. 
We adopt a three-step procedure.
In the first step,  empirical estimates of the covariance functions are constructed. 
Let  $r_{ij}^{(k)} = y_{ij}^{(k)} - \mu^{(k)}\left(t_{ij}^{(k)}\right)$ be the residuals and $C_{i j_1 j_2}^{(kk^{\prime})} = r_{ij_1}^{(k)}r_{ij_2}^{(k^{\prime})}$ be the auxiliary variables. Note that   
$\E\left(C_{i j_1 j_2}^{(kk^{\prime})}\right) = C_{kk^{\prime}}\left(t_{ij_1}^{(k)}, t_{ij_2}^{(k')}\right) + \sigma_k^2 {1}_{\{k=k^{\prime},j_1=j_2\}}$ for $1 \leq j_1\leq m_{ik}, 1\leq j_2\leq m_{ik'}$.
Thus, $C_{ij_1j_2}^{(kk')}$ is an unbiased estimate of $C_{kk'}\left(t_{ij_1}^{(k)}, t_{ij_2}^{(k')}\right) $ whenever $k \neq k^{\prime}$ or $j_1 \neq j_2$.
In the second step, the noisy auxiliary variables are smoothed to obtain smooth estimates of the covariance functions.
For smoothing, we  use bivariate {\it P}-splines \citep{Eilers:03} because  it is an automatic smoother and is computationally simple.
In the final step, we pool all estimates of the individual covariance functions and use an extra step of eigendecomposition 
to obtain refined estimates of covariance functions. The refined estimates lead to a covariance operator that is properly defined, i.e.,  positive semi-definite.
In practice, the mean functions $\mu^{(k)}$s are unknown and we estimate them using  $P$-splines \citep{Eilers:96}
with the smoothing parameters selected by leave-one-subject-out cross validation;
see Appendix A for details.
Denote the estimates by $\widehat{\mu}^{(k)}$.
Let $\widehat{r}_{ij}^{(k)} = y_{ij}^{(k)} - \widehat{\mu}^{(k)}\left(t_{ij}^{(k)}\right)$ and $\widehat{C}_{i j_1 j_2}^{(kk^{\prime})} = \widehat{r}_{i j_1}^{(k)}\widehat{r}_{i j_2}^{(k^{\prime})}$, the actual auxiliary variables.

The bivariate {\it P}-splines model $C_{k k^{\prime}}(s,t)$ uses tensor-product splines $G_{k k^{\prime}}(s,t)$ for $1 \leq k, k^{\prime} \leq p$.
Specifically, $G_{k k^{\prime}}(s,t) = \sum_{1\leq\gamma_1,\gamma_2\leq c}\theta_{\gamma_1\gamma_2}^{(k k^{\prime})} B_{\gamma_1}(s)B_{\gamma_2}(t)$, where $\bTheta_{k k^{\prime}} = \left[\theta_{\gamma_1\gamma_2}^{(k k^{\prime})}\right]_{1\leq\gamma_1,\gamma_2\leq c}\in\mathbb{R}^{c\times c}$
is a coefficient matrix, $\{B_1(\cdot), \ldots, B_c(\cdot)\}$ is the collection of B-spline basis functions in $\T$, and $c$ is the number of equally-spaced interior knots plus the order (degree plus $1$) of the B-splines. 
Because $C_{kk'}(s,t) = C_{k'k}(t,s) = \text{Cov}\left\{x^{(k)}(s), x^{(k')}(t)\right\}$,
it is reasonable to impose
the assumption that 
\begin{equation*}
	\bTheta_{kk'} = \bTheta_{k'k}^{\top}
\end{equation*}
so that $G_{kk'}(s,t) = G_{k'k}(t,s)$. Therefore, in the rest of the section, we consider only $k \leq k'$.

Let $\bD\in\real^{(c-2)\times c}$ denote a second-order differencing matrix such that for a vector $\mathbf{a}=(a_1,\ldots, a_c)^{\top}\in\real^c$,
$\bD\mathbf{a} = (a_3-2a_2+a_1, a_4 -2 a_3 + a_2,\ldots, a_c - 2a_{c-1} + a_{c-2})^{\top}\in\real^{c-2}$.
Also let $\|\cdot\|_F$ be the Frobenius norm.
For the cross-covariance function $C_{kk^{\prime}}(s,t)$ with $k < k^{\prime}$, 
the bivariate {\it P}-splines estimate the coefficient matrix $\bTheta_{kk'}$ by $\widehat{\bTheta}_{kk'}$ which minimizes the penalized least squares
\begin{eqnarray}\label{eq:cross-cov}
	\sum_{i=1}^n \sum_{1 \leq j_1\leq m_{ik}} \sum_{ 1\leq  j_2 \leq m_{ik'}} \left\{G_{kk^{\prime}}\left(t_{ij_1}^{(k)}, t_{ij_2}^{(k')}\right) - \widehat{C}_{i j_1 j_2}^{(kk^{\prime})}\right\}^2 + \lambda_{kk^{\prime}1} \|\bD\bTheta_{kk^{\prime}}\|_F^2 + \lambda_{kk^{\prime}2} \|\bD\bTheta_{kk^{\prime}}^{\top}\|_F^2,
\end{eqnarray} 
where 
$\lambda_{kk^{\prime}1}$ and $\lambda_{kk^{\prime}2}$ are two nonnegative smoothing parameters that balance the model fit and smoothness of the estimate
and will be determined later.
Indeed, the column penalty $\|\bD\bTheta_{kk'}\|_F^2$ penalizes the 2nd order consecutive differences of the columns of $\bTheta_{kk'}$
and similarly, the row penalty $\|\bD\bTheta_{kk'}^{\top}\|_F^2$ penalizes the 2nd order consecutive differences of the rows of $\bTheta_{kk'}$.
The two penalty terms are essentially penalizing the 2nd order partial derivatives of $G_{kk'}(s,t)$ along the $s$ and $t$ directions, respectively.
The two smoothing parameters are allowed to differ to accommodate different levels of smoothing along the two directions.

For the auto-covariance functions $C_{kk}(s,t)$ with $k = 1, \ldots, p$, 
we conduct bivariate covariance smoothing by enforcing the following constraint on the coefficient matrix $\bTheta_{kk}$ \citep{Xiao:16b},
\begin{equation}\label{eq:symm}
	\bTheta_{kk} = \bTheta_{kk}^{\top}.
\end{equation}
It follows that $G_{kk}(s,t)$ is a symmetric function.
Then the coefficient matrix $\bTheta_{kk}$ and the error variance $\sigma_k^2$ are jointly estimated by $\widehat{\bTheta}_{kk}$ and $\widehat{\sigma}_k^2$,
which minimize the penalized least squares
\begin{eqnarray}\label{eq:auto-cov}
	\sum_{i=1}^n  \sum_{1 \leq j_1, j_2 \leq m_{ik}}\left\{G_{kk}\left(t_{ij_1}^{(k)}, t_{ij_2}^{(k)}\right) + \sigma_k^2 {1}_{\{j_1=j_2\}} - \widehat{C}_{i j_1 j_2}^{(kk)}\right\}^2 + \lambda_k 
	\|\bD\bTheta_{kk}\|_F^2,
\end{eqnarray}
over all symmetric $\bTheta_{kk}$ and $\lambda_k$ is a smoothing parameter. Note that the two penalty terms in \eqref{eq:cross-cov}
become the same when $\bTheta_{kk}$ is symmetric and thus only one smoothing parameter is needed for auto-covariance estimation.

\subsubsection{Estimation}


We first introduce the notation. 
Let $\ve(\cdot)$ be an operator that
stacks the columns of a matrix into a column vector
and denote by $\otimes$ the Kronecker product.
Fix $k$ and $k'$ with $k\leq k'$.
Let $\btheta_{kk^{\prime}} = \ve(\bTheta_{kk^{\prime}})\in \real^{c^2}$ be a vector of the coefficients
and $\bb(t) = \{B_1(t), \ldots, B_c(t)\}^{\top}\in \real^c$ denotes the B-spline base.
Then  $$G_{k k^{\prime}}(s,t) = \bb(s)^{\top}\bTheta_{kk'}\bb(t) = \{\bb(t) \otimes \bb(s)\}^{\top} \btheta_{kk^{\prime}}.$$
We now organize the auxiliary responses $\widehat{C}_{ij_1j_2}^{(kk')}$ for each pair of $k$ and $k'$. Let $\mathbf{r}_i^{(k)} = \left(r_{i1}^{(k)},\ldots, r_{im_{ik}}^{(k)}\right)^{\top}\in\real^{m_{ik}}$,
$\bChat_{i}^{(kk^{\prime})} 
= \mathbf{r}_i^{(k)} \otimes \mathbf{r}_i^{(k')}\in \real^{m_{ik}m_{ik'}}$
and 
$\bChat^{(kk^{\prime})} = \left(\bChat_{1}^{(kk^{\prime}), \top},\ldots, \bChat_{n}^{(kk^{\prime}), \top}\right)^{\top}\in \real^{N_{kk'}}$, 
where $N_{kk'} = \sum_{i=1}^n m_{ik} m_{ik'}$
is the total number of auxiliary responses for the pair of $k$ and $k'$.
As for the B-splines,
let $\bb_{i}^{(k)} =\left [\bb\left(t_{i1}^{(k)}\right),\ldots, \bb\left(t_{im_{ik}}^{(k)}\right)\right]\in \real^{c \times m_{ik}}$,
$\bB_{i}^{(kk')} = \left(\bb_i^{(k')}\otimes \bb_i^{(k)}\right)^{\top}
\in \real^{(m_{ik}m_{ik'})\times c^2}$, and
$\bB^{(kk')} = \left[\bB_1^{(kk'),\top},\ldots, \bB_n^{(kk'),\top}\right]^{\top}\in \real^{N_{kk'}\times c^2}$.

For  estimation of the cross-covariance functions $C_{kk'}$ with $k < k^{\prime}$,
the penalized least squares in \eqref{eq:cross-cov} can be rewritten as
\begin{eqnarray}\label{eq:cross-covPLS}
	\left(\bChat^{(kk^{\prime})} - \bB^{(kk')}\btheta_{kk^{\prime}} \right)^{\top}\left(\bChat^{(kk^{\prime})} - \bB^{(kk')}\btheta_{kk^{\prime}}\right) + \lambda_{kk^{\prime}1} \btheta_{kk^{\prime}}^{\top}\bP_1\btheta_{kk^{\prime}} + \lambda_{kk^{\prime}2} \btheta_{kk^{\prime}}^{\top}\bP_2\btheta_{kk^{\prime}},
\end{eqnarray} 
where $\bP_1 = \bI_c \otimes \bD^{\top}\bD$ and $\bP_2 = \bD^{\top}\bD \otimes  \bI_c$. 
The expression in \eqref{eq:cross-covPLS} is a quadratic function of the coefficient vector $\btheta_{kk'}$.
Therefore,  we derive that
$$\widehat{\btheta}_{kk^{\prime}} = \left(\bB^{(kk'), \top}\bB^{(kk')} + \lambda_{kk^{\prime}1}\bP_1 + \lambda_{kk^{\prime}2} \bP_2\right)^{-1}\bB^{(kk'), \top}\bChat^{(kk^{\prime})}$$
and the estimate of the cross-covariance function $C_{kk'}(s,t)$ is $\widehat{C}_{kk'}(s,t) = \{\bb(t) \otimes \bb(s)\}^{\top} \widehat{\btheta}_{kk^{\prime}}$.

For estimation of the auto-covariance functions, because of the constraint on the coefficient matrix in~\eqref{eq:symm}, 
let $\bfeta_k \in \real^{c(c+1)/2}$ be a vector obtained by stacking the columns of the lower triangle of $\bTheta_{kk}$
and let $\bG_c\in \real^{c^2\times c(c+1)/2}$ be a duplication matrix  such that $\btheta_{kk} = \bG_c \bfeta_k$ (Page 246, \citealt{Seber:07}). 
Let $\bZ_{i}^{(k)} = \ve(\bI_{m_{ik}})\in \real^{m_{ik}^2}$ and $\bZ^{(k)} = \left(\bZ_1^{(k), \top},\ldots, \bZ_n^{(k), \top}\right)^{\top}\in\real^{N_{kk}}$.
Finally let $\bbeta_{k} = \left(\bfeta_k^{\top},\sigma_k^2\right)^{\top} \in \real^{\tilde{c}}$ with $\tilde{c} = c(c+1)/2 + 1$.
It follows that the penalized least squares in \eqref{eq:auto-cov} can be rewritten as
\begin{eqnarray*}
	\left(\bChat^{(kk)} - \bX^{(k)} \bbeta_{k}\right)^{\top}\left(\bChat^{(kk)}  - \bX^{(k)} \bbeta_{k}\right) + \lambda_k \bbeta_{k}^{\top}\bQ\bbeta_{k},
\end{eqnarray*}
where $\bX^{(k)} = \left[\bB^{(kk)}, \bZ^{(k)}\right] \in \real^{N_{kk} \times \tilde{c}}$
and $\bQ  = \text{blockdiag}\left\{\bG_c^{\top}(\bI_c \otimes \bD\bD^{\top})\bG_c^{\top}, 0\right\}\in \real^{\tilde{c}\times \tilde{c}}$. 
Therefore, we obtain
$$\widehat{\bbeta}_k = \left(\widehat{\bfeta}_k^{\top}, \widehat{\sigma}_k^2\right)^{\top} = \left(\bX^{(k), \top}\bX^{(k)} + \lambda_k \bQ\right)^{-1}\bX^{(k), \top}\bChat^{(kk)}.$$
It follows that $\widehat{\btheta}_{kk} = \bG_c \widehat{\bfeta}_k$
and the estimate of the auto-covariance function $C_{kk}(s,t)$ is $\widehat{C}_{kk}(s,t) = \{\bb(t) \otimes \bb(s)\}^{\top} \widehat{\btheta}_{kk}$.

The above estimates of covariance functions may not lead to a positive semi-definite covariance operator and thus
have to be refined.
We pool all estimates together and we shall use the following proposition.


\begin{prop} \label{eq:eigen}
	Assume that $C_{kk^{\prime}}(s,t) = \bb(s)^{\top} \bTheta_{k k^{\prime}} \bb(t)$.
	Let $\bG = \int \bb(t) \bb(t)^{\top} dt \in \real^{c\times c}$
	and
	assume that  $\bG$ is positive definite \citep{zhou1998local}. 
	Then $ \left[\bG^{\frac{1}{2}} \bTheta_{k k^{\prime}} \bG^{\frac{1}{2}}\right]_{1 \leq k, k^{\prime} \leq p} \in \real^{pc \times pc}$ admits the spectral decomposition, 
	$\sum_{\ell=1}^{\infty} d_{\ell} \bu_{\ell}\bu_{\ell}^{\top}$,
	where $d_{\ell}$ is the $\ell$th largest eigenvalue of the covariance operator $\bGamma$, and $\bu_{\ell} = \left\{\bu_{\ell}^{(1), \top}, \ldots, \bu_{\ell}^{(p),\top}\right\}^{\top}\in\real^{pc}$ is the associated eigenvector with $\bu_{\ell}^{(k)}\in \real^c$ and such that $\Psi_{\ell}^{(k)}(t) = \bb(t)^{\top}\bG^{-\frac{1}{2}}\bu_{\ell}^{(k)}$.
\end{prop}

The proof is provided in Appendix B.
Proposition \ref{eq:eigen} implies that, with the tensor-product B-spline representation of the covariance functions,
one spectral decomposition  gives us the eigenvalues and eigenfunctions.
In particular, the eigenfunctions $\Psi_{\ell}^{(k)}(t)$ are linear combinations of the B-spline basis functions, which means that
they can be straightforwardly evaluated, an
advantage of spline-based methods compared to other smoothing methods for which eigenfunctions are approximated by
spectral decompositions of the covariance functions evaluated at a grid of time points.

Once we have $\widehat{\bTheta}_{k k^{\prime}}$, the estimate of the coefficient matrix $\bTheta_{kk'}$,
the spectral decomposition of $[\bG^{\frac{1}{2}} \widehat{\bTheta}_{k k^{\prime}} \bG^{\frac{1}{2}}]_{kk'}$
gives us  estimates $\widehat{d}_{\ell}$ and $\widehat{\bu}_{\ell} = \left\{\widehat{\bu}_{\ell}^{(1), \top}, \ldots, \widehat{\bu}_{\ell}^{(p),\top}\right\}^{\top}$.
We discard negative $\widehat{d}_{\ell}$ to ensure that the multivariate covariance operator is positive semi-definite and this leads
to a refined estimate of the coefficient matrix $\bTheta_{kk'}$, $\bThetatilde_{kk'} = \bG^{-\frac{1}{2}}\left\{\sum_{\ell: \widehat{d}_{\ell} > 0} \widehat{d}_{\ell} \widehat{\bu}_{\ell}^{(k)}\widehat{\bu}_{\ell}^{(k'),\top}\right\}\bG^{-\frac{1}{2}}$.
Then the refined estimate of the covariance functions is
$\widetilde{C}_{kk'}(s,t) = \bb(s)^{\top}\bThetatilde_{kk'}\bb(t)$.
Proposition \ref{eq:eigen} also suggests that the eigenfunctions can be estimated by
$\widetilde{\Psi}_{\ell}^{(k)}(t) = \bb(t)^{\top}\bG^{-\frac{1}{2}}\widehat{\bu}_{\ell}^{(k)}$.

For principal component analysis or curve prediction in practice, 
one may select further the number of principal components 
by either the proportion of variance explained (PVE) \citep{Greven:10} or an AIC-type criterion \citep{li2013selecting}. Here, we follow \cite{Greven:10} using PVE with a value of $0.99$.

\subsubsection{Selection of Smoothing Parameters}
We select the smoothing parameters in each auto-covariance/cross-covariance estimation
using leave-one-subject-out cross-validation; see, e.g., \cite{Yao:05a} and \cite{Xiao:16b}. 
A fast approximate algorithm for the auto-covariance has been derived in \cite{Xiao:16b}. 
So we focus on the cross-covariance and use the notation in \eqref{eq:cross-covPLS}.
Note that there are two smoothing parameters for each cross-covariance estimation.

For simplicity, we suppress the superscript and subscript $kk^{\prime}$ in \eqref{eq:cross-covPLS} for both $\bChat$ and $\bB$.
Let $\bCtilde_i^{[i]}$ be the prediction of the auxiliary responses $\bChat_i$ from the estimate using data
without the $i$th subject. Let $\|\cdot\|$ be the Euclidean norm and  the cross-validation error is
\begin{eqnarray}\label{eq:icv}
	\textnormal{iCV} = \sum_{i=1}^n\left\|\bChat_i - \bCtilde_i^{[i]}\right\|^2.
\end{eqnarray}
We shall also now suppress  the subscript $k$ from $m_{ik}$ and $kk'$ from  $N_{kk'}$.
Let $\bS= \bB(\bB^{\top}\bB + \lambda_1 \bP_1 + \lambda_2 \bP_2)^{-1}\bB^{\top} \in \mathbb{R}^{N\times N}$,
$\bS_i = \bB_i(\bB^{\top}\bB + \lambda_1 \bP_1 + \lambda_2 \bP_2)^{-1}\bB^{\top} \in \mathbb{R}^{m_{i}^2\times N}$,
and
$\bS_{ii} = \bB_i(\bB^{\top}\bB + \lambda_1 \bP_1 + \lambda_2 \bP_2)^{-1}\bB_i^{\top}\in \mathbb{R}^{m_i^2\times m_i^2}$. 
Then a short-cut formula for \eqref{eq:icv} is
$$
\textnormal{iCV} = \sum_{i=1}^n \left\|\left(\bI_{m_i^2} - \bS_{ii}\right)^{-1}(\bS_i\bChat -\bChat_i) \right\|^2.
$$
Similar to \cite{Xu:12} and \cite{Xiao:16b}, the iCV can be further simplified by adopting the approximation
$(\bI_{m_i^2} - \bS_{ii})^{-2} = \bI_{m_i^2} + 2\bS_{ii}$, which results in the generalized cross validation,
denoted by iGCV, 
\begin{equation}\label{eq:igcv_a}
	\iGCV 
	= \left\|\bChat-\bS\bChat\right\|^2 + 2\sum_{i=1}^n \left(\bS_i\bChat - \bChat_i\right)^{\top} \bS_{ii}\left(\bS_i\bChat - \bChat_i\right).
\end{equation}
While iGCV is much easier to compute than iCV, the formula in~\eqref{eq:igcv_a} is still computationally expensive to compute. Indeed, the smoother matrix $\bS$ is of dimension $2500\times 2500$ if $n=100$ and $m_i = m =5$ for all $i$. Thus, we need to further simplify the formula.

Let $\bG_n = \bB^{\top}\bB$,
$\bBtilde=\bB\bG_n^{-1/2}\in\mathbb{R}^{N\times c^2}$, $\bBtilde_i=\bB_i\bG_n^{-1/2}\in\mathbb{R}^{m_i^2\times c^2}$,
$\bbf=\bBtilde^{\top}\bChat\in\mathbb{R}^{c^2}$, $\bbf_i=\bBtilde_i^{\top}\bChat_i\in\mathbb{R}^{c^2}$,
and $\bL_i = \bBtilde_i^{\top}\bBtilde_i\in\mathbb{R}^{c^2\times c^2}$.
Also let $\bPtilde_1 =\bG_n^{-1/2}\bP_1 \bG_n^{-1/2}\in\mathbb{R}^{c^2\times c^2}$,
$\bPtilde_2 =\bG_n^{-1/2}\bP_2 \bG_n^{-1/2}\in\mathbb{R}^{c^2\times c^2}$,
and $\bSigma = \bI_{c^2} + \lambda_1\bPtilde_1 + \lambda_2\bPtilde_2$.
Then \eqref{eq:igcv_a} can be simplified as
\begin{equation}\label{eq:igcv_a1}
	\iGCV  =\left\|\bChat\right\|^2 - 2\bbf^{\top}\bSigma^{-1}\bbf + \bbf^{\top}\bSigma^{-2}\bbf
	+ 2\sum_{i=1}^n\left(\bL_i\bSigma^{-1}\bbf-\bbf_i\right)^{\top}\bSigma^{-1}\left(\bL_i\bSigma^{-1}\bbf-\bbf_i\right).
\end{equation}
Note that $\bSigma$ has two smoothing parameters. Following \cite{Wood:00}, we use an equivalent parameterization
$\bSigma = \bI + \rho \{w \bPtilde_1 + (1-w)\bPtilde_2\}$, where $\rho=\lambda_1 + \lambda_2$ represents the overall smoothing level
and $w = \lambda_1\rho^{-1}\in [0,1]$ is the relative weight of $\lambda_1$. 
We conduct a two-dimensional grid  search of $(\rho,w)$ as follows. For a given $w$, let $\bU\text{diag}(\bs)\bU^{\top}$ be the eigendecompsition of $w \bPtilde_1 + (1-w)\bPtilde_2$,
where $\bU\in\mathbb{R}^{c^2\times c^2}$ is an orthonormal matrix and $\bs = (s_1,\ldots, s_{c^2})\in\mathbb{R}^{c^2}$ is the vector of eigenvalues.
Then $\bSigma^{-1} = \bU\text{diag}(\bdtilde)\bU^{\top}$ with $\bdtilde =1/(1+\rho\bs)\in\mathbb{R}^{c^2}$.
\begin{prop} \label{eq:igcv}
	Let $\odot$ stand for the point-wise multiplication. Then,
	\[
	\iGCV = \left\|\bChat\right\|^2 + (\bbftilde\odot\bdtilde)^{\top}(\bbftilde\odot\bdtilde) - 2\bdtilde^{\top}\bg   - 4\bdtilde^{\top} \bF \bdtilde + 2 \bdtilde^{\top}\left[\sum_{i=1}^n \left\{\bLtilde_i(\bbftilde\odot\bdtilde)\right\} \odot \left\{\bLtilde_i(\bbftilde\odot\bdtilde)\right\}\right],
	\]
	where $\bbftilde_i = \bU^{\top}\bbf_i\in\mathbb{R}^{c^2}$, $\bbftilde=\bU^{\top}\bbf\in\mathbb{R}^{c^2}$,
	$\bg = \bbftilde\odot\bbftilde - \sum_{i=1}^n \bbftilde_i\odot\bbftilde_i\in\mathbb{R}^{c^2}$, 
	$\bLtilde_i = \bU^{\top}\bL_i\bU\in\mathbb{R}^{c^2\times c^2}$,
	and $\bF = \sum_{i=1}^n(\bbftilde_i{\bbftilde}^{\top})\odot \bLtilde_i\in\mathbb{R}^{c^2\times c^2}$.
\end{prop}
The proof is provided in Appendix B. 
For each $w$, note that only $\bdtilde$ depends on $\rho$
and needs to be calculated repeatedly,  and
all other terms need to be calculated only once.
The entire algorithm is presented in Algorithm \ref{algo:igcv}.
We give an evaluation of the complexity of the proposed algorithm.
Assume that $m_i = m$ for all $i$.
The first initialization (step 1) requires $O(nm^2 c^2 + n c^4 + c^6)$ computations.
For each $w$, the second initialization (step 3) also requires $O\{nc^4\min(m^2, c^2) + c^6\}$ computations.
For each $\rho$, steps 5-10 requires $O(nc^4)$ computations. 
Therefore, the formula in Proposition \ref{eq:igcv} is most efficient to calculate for sparse data with small numbers of observations per subject, 
i.e., $m_i$s are small.

\begin{algorithm}
	\SetKwInput{Input}{Input}
	\SetKwInput{Output}{Output}
	\Input{$\bB$, $\bChat$, $\bP_1$, $\bP_2$, $\brho = \{\rho_{1}, \ldots, \rho_{K}\}^{\top}$, $\bw = \{w_{1}, \ldots, w_{K}\}^{\top}$}
	\Output{$\{\rho^*, w^*\}$}
	Initialize $\|\bChat\|^2$, $\bPtilde_1$, $\bPtilde_2$, $\bG_n^{-1/2}$, $\bBtilde$, $\bBtilde_i$, $\bbf$,  $\bbf_i$, $\bL_i$ for $i=1,\ldots, n$;\\
	\ForEach{$w$ in $\bw$}{
		Initialize $\bs$, $\bbftilde$, $\bbftilde_i$, $\bg$, $\bF$, $\bLtilde_i$, $i = 1, \ldots, n$\;
		\ForEach{$\rho$ in $\brho$}{
			$\bdtilde \gets 1/(1+ \rho\bs)$\;
			$I \gets (\bbftilde\odot\bdtilde)^{\top}(\bbftilde\odot\bdtilde)$\;
			$II \gets - 2\bdtilde^{\top}\bg$\;
			$III \gets  - 4\bdtilde^{\top} \bF \bdtilde$\;
			$IV \gets 2 \bdtilde^{\top}\left[\sum_{i=1}^n \left\{\bLtilde_i(\bbftilde\odot\bdtilde)\right\} \odot \left\{\bLtilde_i(\bbftilde\odot\bdtilde)\right\}\right]$\;
			$\iGCV \gets \|\bChat\|^2 + I + II + III + IV$\;
		}
	}
	$\{\rho^*, w^*\} \gets \arg\min_{\rho, w}\iGCV$\;
	\caption{Selection of smoothing parameters}
	\label{algo:igcv}
\end{algorithm}

\subsection{Prediction} \label{sec:prediction}

For prediction, assume that the smooth curve $\bx_i(t)$ is generated from a multivariate Gaussian process. 
Suppose that we want to predict the $i$th multivariate response
$\bx_i(t)$ at $\{s_{i1},\ldots, s_{im}\}$ for $m\geq 1$.
Let $\by_i^{(k)} =  \left(y_{i1}^{(k)},\ldots,y_{im_{ik}}^{(k)}\right)^{\top}$
be the vector of observations at 
$\left\{t_{i1}^{(k)},\ldots, t_{im_{ik}}^{(k)}\right\}$ for the $k$th response.
Let $\bmu_i^{(k), o} = \left(\mu^{(k)}\left(t_{i1}^{(k)}\right),\ldots,\mu^{(k)}\left(t_{im_{ik}}^{(k)}\right)\right)^{\top}$ be the vector of the $k$th mean function at the observed time points. Let $\by_i = \left(
\by_i^{(1),\top},\ldots, \by_i^{(p),\top}\right)^{\top}$
and $\bmu_i^{o} = \left(
\bmu_i^{(1), o,\top},\ldots, \bmu_i^{(p), o,\top}
\right)^{\top}$.
Let $\bmu_i^n = \left(\mu^{(1)}(s_{i1}),\ldots,\mu^{(1)}(s_{im}), \cdots, \mu^{(p)}(s_{i1}),\ldots,\mu^{(p)}(s_{im})\right)^{\top}$
be the vector of mean functions at the time points for prediction.
It follows that
\[
\left(\begin{array}{c}\by_i \\\bx_i\end{array}\right) \sim \mathcal{N}\left\{
\left(\begin{array}{c}\bmu_i^o \\ \bmu_i^n\end{array}\right),
\left(\begin{array}{cc}\C(\by_i)& \C(\bx_i, \by_i)^{\top}  \\\C(\bx_i, \by_i)
& \C(\bx_i)  \end{array}\right) 
\right\}.
\]
Thus, we obtain
\[
\E(\bx_i|\by_i) = \C(\bx_i, \by_i)\C(\by_i)^{-1} (\by_i - \bmu_i^o) + \bmu_i^n,
\]
\[
\C(\bx_i|\by_i) = \C(\bx_i) - \C(\bx_i, \by_i)\C(\by_i)^{-1}\C(\bx_i, \by_i)^{\top},
\]
Let $\bb_i^{(k), o} = \left[\bb\left(t_{i1}^{(k)}\right),\ldots, \bb\left(t_{im_{ik}}^{(k)}\right)\right]^{\top}$
and
$\bb_i^{n} = \left[\bb\left(s_{i1}\right),\ldots, \bb\left(s_{im}\right)\right]^{\top}$. 
Next let $\bB_i^{o} = \text{blockdiag}\left(\bb_i^{(1),o},\ldots,
\bb_i^{(p),o}\right)$ and $\bB_i^{n} = \bI_p\otimes \bb_i^{n}$.
Then $\C(\by_i)$ given by $\bB_i^o\bTheta \bB_i^{o, \top} +  \text{blockdiag}(\sigma_1^2\bI_{m_{i1}}, \ldots, \sigma_p^2\bI_{m_{ip}}) $, and $\C(\bx_i)$ and $\C(\by_i, \bx_i)$ are given  by $\bB_i^n\bTheta \bB_i^{n, \top}$ and $\bB_i^o\bTheta \bB_i^{n,\top}$, respectively.
Let $\bThetatilde = \left[\bThetatilde_{k k^{\prime}}\right]_{1 \leq k, k^{\prime} \leq p} \in \real^{pc \times pc}$.
Plugging in the estimates, 
we predict $\bx_i$ by
\begin{eqnarray*}
	\widehat{\bx}_i &=&  \left(\widehat{x}_{i}^{(1)}(s_{i1}),\ldots,\widehat{x}_{i}^{(1)}(s_{im}), \cdots, \widehat{x}_{i}^{(p)}(s_{i1}),\ldots,\widehat{x}_{i}^{(p)}(s_{im})\right)^{\top} \\
	&=& \left(\bB_i^n\widetilde{\bTheta}\bB_i^{o,\top}\right) \widehat{\bV}_i^{-1} (\by_i - \widehat{\bmu}_i^o) + \widehat{\bmu}_i^n,
\end{eqnarray*}
where 
$\widehat{\bmu}_i^o = \left(\widehat{\mu}^{(1)}\left(t_{i1}^{(1)}\right),\ldots,\widehat{\mu}^{(1)}
\left(t_{im_{i1}}^{(1)}\right), \cdots, \widehat{\mu}^{(p)}\left(t_{i1}^{(p)}\right),\ldots,\widehat{\mu}^{(p)}
\left(t_{im_{ip}}^{(p)}\right)\right)^{\top}$ is the estimate of $\bmu_i^o$,
$\widehat{\bmu}_i^n = \left(\widehat{\mu}^{(1)}(s_{i1}),\ldots,\widehat{\mu}^{(1)}(s_{im}), \cdots, \widehat{\mu}^{(p)}(s_{i1}),\ldots,\widehat{\mu}^{(p)}(s_{im})\right)^{\top}$ is the estimate of $\bmu_i^n$, $\widehat{\bV}_i = \bB_i^o\widetilde{\bTheta}\bB_i^{o,\top} + \text{blockdiag}\left(\widehat{\sigma}_1^2\bI_{m_{i1}}, \ldots, \widehat{\sigma}_p^2\bI_{m_{ip}}\right)$.
An approximate covariance matrix for $\widehat{\bx}_i$ is
$$
\widehat{\C}(\bx_i|\by_i) = \bB_i^n\widetilde{\bTheta}\bB_i^{n,\top} - \left(\bB_i^n\widetilde{\bTheta}\bB_i^{o,\top}\right) \widehat{\bV}_i^{-1} \left(\bB_i^n\widetilde{\bTheta}\bB_i^{o,\top}\right)^{\top}.
$$
Therefore, a $95\%$ point-wise confidence interval for the $k$th response is given by
\[
\widehat{x}_i^{(k)}(s_{ij}) \pm 1.96 \sqrt{\widehat{\V}\left(x_i^{(k)}(s_{ij})|\by_i\right)},
\]
where $\widehat{\V}\left(x_i^{(k)}(s_{ij})|\by_i\right)$ can be extracted from the diagonal of $\widehat{\C}(\bx_i|\by_i)$.

Finally, we predict the first $L\geq 1$ scores $\bxi_i = (\xi_{i1},\ldots, \xi_{iL})^{\top}$ for the $i$th subject.
Note that $\xi_{i\ell} = \int  \bPsi_{\ell}(t)^{\top}\{\bx_i(t) - \bmu(t)\}dt$.
With a similar derivation as above,
$\bx_i(t) - \bmu(t)$ can be predicted by
$\left\{\bI_p\otimes \bb(t)\right\}^{\top}\widetilde{\bTheta}\bB_i^{o,\top}\widehat{\bV}_i^{-1} (\by_i - \widehat{\bmu}_i^o)$.
By Proposition \ref{eq:eigen}, the eigenfunctions
$\Psi_{\ell}^{(k)}(t)$ are estimated by $\bb(t)^{\top}\bG^{-\frac{1}{2}}\widehat{\bu}_{\ell}^{(k)}$ and thus
$\bPsi_{\ell}(t)^{\top} = \widehat{\bu}_{\ell}^{\top}\{\bI_p\otimes \bG^{-\frac{1}{2}}\bb(t)\}$.
It follows that $$\widehat{\xi}_{i\ell} = \widehat{\bu}_{\ell}^{\top}\left(\bI_p \otimes \bG^{\frac{1}{2}}\right)\widetilde{\bTheta}\bB_i^{o,\top}\widehat{\bV}_i^{-1} (\by_i - \widehat{\bmu}_i^o).$$
\section{Simulations} \label{sec:simulation}
We evaluate  the finite sample performance of the proposed method (denoted by mFACEs) against mFPCA via a synthetic simulation study
and a simulation study mimicking the ADNI data in the real data example. Here, we report the details and results of the former as the conclusions
remain the same for the latter and details are provided in the supplement.

\subsection{Simulation Settings and Evaluation Criteria}
We generate data by model~\eqref{eq:model}
with $p=3$ responses.
The mean functions are $\bmu(t) = [5 \sin (2 \pi t), 5 \cos (2 \pi t), 5(t-1)^2]^{\top}$.
We first specify the auto-covariance functions. Let $\bPhi_1(t) = \left[\sqrt{2} \sin (2 \pi t), \sqrt{2} \cos (4 \pi t), \sqrt{2} \sin (4\pi t)\right]^{\top}$, 
$\bPhi_2(t) = \left[\sqrt{2} \cos (\pi t), \sqrt{2} \cos (2\pi t), \sqrt{2} \cos (3 \pi t)\right]^{\top}$, and $\bPhi_3(t) = \left[\sqrt{2} \sin (\pi t), \sqrt{2} \sin (2\pi t), \sqrt{2} \sin (3 \pi t)\right]^{\top}$. Also let
\begin{eqnarray*}
	\bLambda_{11} = \begin{pmatrix} 3 & 0 & 0 \\ 0 & 1.5 & 0 \\ 0 & 0 & 0.75 \end{pmatrix}, \,\, \bLambda_{22} = \begin{pmatrix} 3.5 & 0 & 0 \\ 0 & 1.75 & 0 \\ 0 & 0 & 0.5 \end{pmatrix}, \,\, \bLambda_{33} = \begin{pmatrix} 2.5 & 0 & 0 \\ 0 & 2 & 0 \\ 0 & 0 & 1 \end{pmatrix}.
\end{eqnarray*}
Then the auto-covariance functions are $C_{kk}(s, t) = \bPhi_k(s)^{\top}\bLambda_{kk}\bPhi_k(t), k=1, 2, 3$.
For the cross-covariance functions,  let
$C_{kk'}(s,t) = \rho\bPhi_k(s)^{\top}\bLambda_{kk}^{\frac {1}{2}}\bLambda_{k'k'}^{\frac{1}{2}}\bPhi_{k'}(t)$ for $k\neq k'$,
where $\rho\in [0,1]$ is a parameter to be specified. 
The induced covariance operator from the above specifications is proper; see Lemma \ref{lem1} in Appendix C.
It is easy to derive that the absolute value of cross-correlation $\rho_{kk'}(s,t) = C_{kk'}(s,t)/\sqrt{C_{kk}(s,s) C_{k'k'}(t,t)}$ is bounded by $\rho$. 
Hence, $\rho$ controls the overall level of correlation between responses: if $\rho = 0$, then the responses are uncorrelated from each other.
The eigendecomposition of the multivariate covariance function gives 9 non-zero eigenvalues with associated multivariate eigenfunctions, 
hence, for $\ell=1,\ldots, 9$,
we  simulate  the scores $\xi_{i \ell}$ from $\mathcal{N}(0, d_{\ell})$,
where $d_{\ell}$ are the induced eigenvalues. 
Next, we simulate the white noises $\epsilon^{(k)}_{ij}$ from $\mathcal{N}(0,\sigma_{\epsilon}^2)$,
where $\sigma_{\epsilon}^2$ is determined according to the signal-to-noise ratio $\text{SNR} = \sum_{\ell} d_{\ell} / (p \sigma_{\epsilon}^2)$. Here, we let $\text{SNR} = 2$.
For each response, the sampling time points are drawn from a uniform distribution in the unit interval and
the number of observations for each subject, $m_{ik}$,  is generated from a uniform discrete distribution on $\{3,4,5,6,7\}$.
Thus, the sampling points not only vary from subject to subject but also vary across responses within each subject.

We use a factorial design with  two factors: the number of subjects $n$ and the correlation parameter $\rho$.
We let $n=100, 200$ or $400$. We let $\rho = 0.5$, which corresponds to a weak correlation between responses as
the average absolute correlations between responses is only $0.36$. Another value of $\rho$ is 0.9, which
corresponds to a moderate correlation between responses as the average absolute correlations between responses is about $0.50$. 
In total, we have 6 model conditions and for each model condition we generate $200$ datasets.
To evaluate the prediction accuracy of the various methods, we draw $200$ additional subjects as testing data.
The true correlation functions and a sample of the simulated data are shown in the supplement.

We compare mFACEs and mFPCA in terms of estimation accuracy of the covariance functions, the eigenfunctions and eigenvalues,
and prediction of new subjects.
For covariance function estimation, we use the relative integrated square errors (RISE). Let $\widehat{C}_{kk^{\prime}}(s,t)$ be an estimate of $C_{kk^{\prime}}(s,t)$, then RISE are given by
$$
\frac{\sum_{k=1}^p\sum_{k^{\prime}=1}^p\int_0^1 \int_0^1 \left\{C_{kk^{\prime}}(s,t) - \widehat{C}_{kk^{\prime}}(s,t)\right\}^2 ds dt}{\sum_{k=1}^p\sum_{k^{\prime}=1}^p\int_0^1 \int_0^1 \left\{C_{kk^{\prime}}(s,t)\right\}^2 ds dt}.
$$
For estimating the $\ell$th eigenfunction, we use the integrated square errors (ISE), which are defined as
$$
\min \left[\sum_{k=1}^{p}\int_0^1\left\{\Psi_{\ell}^{(k)}(t) - \widehat{\Psi}_{\ell}^{(k)}(t)\right\}^2dt, \,\, \sum_{k=1}^{p}\int_0^1\left\{\Psi_{\ell}^{(k)}(t) + \widehat{\Psi}_{\ell}^{(k)}(t)\right\}^2dt\right].
$$ 
Note that the range of ISE is $[0,2]$.
For estimating the eigenvalues, we use the ratio of the estimate against the truth, i.e., $\hat{d}_{\ell} / d_{\ell}$.
For predicting new curves, we use the mean integrated square errors (MISE), which are given by
$$
\frac{1}{200 p}\sum_{k=1}^p\sum_{i=1}^{200}\left[\int_0^1\left\{x_i^{(k)}(t) - \widehat{x}_i^{(k)}(t) \right\}^2 dt\right].
$$

For the curve prediction using mFPCA, we truncate the number of principal components using a PVE of 0.99.
It is worth noting that if no truncation is adopted, then the curve prediction using mFPCA reduces to curve prediction
using univariate FPCA.
We shall also consider the conditional expectation method based on the estimates
of covariance functions from mFPCA. 
The method is denoted by mFPCA(CE) and its difference with mFACEs is that
different estimates of covariance functions are used.
%
\subsection{Simulation Results}
Figure~\ref{fig:cov_sim2} 
gives boxplots of RISEs of mFACEs and mFPCA for estimating covariance functions.
Under all model conditions, mFACEs outperforms mFPCA and the improvement in RISEs as the sample
size increases is much more pronounced for mFACEs. 
Under the model conditions with moderate correlations ($\rho = 0.9$), the advantage of mFACEs is substantial 
even for the small sample size $n = 100$.

\begin{figure}[htp]
	\centering
	\scalebox{0.35}{
		\includegraphics[]{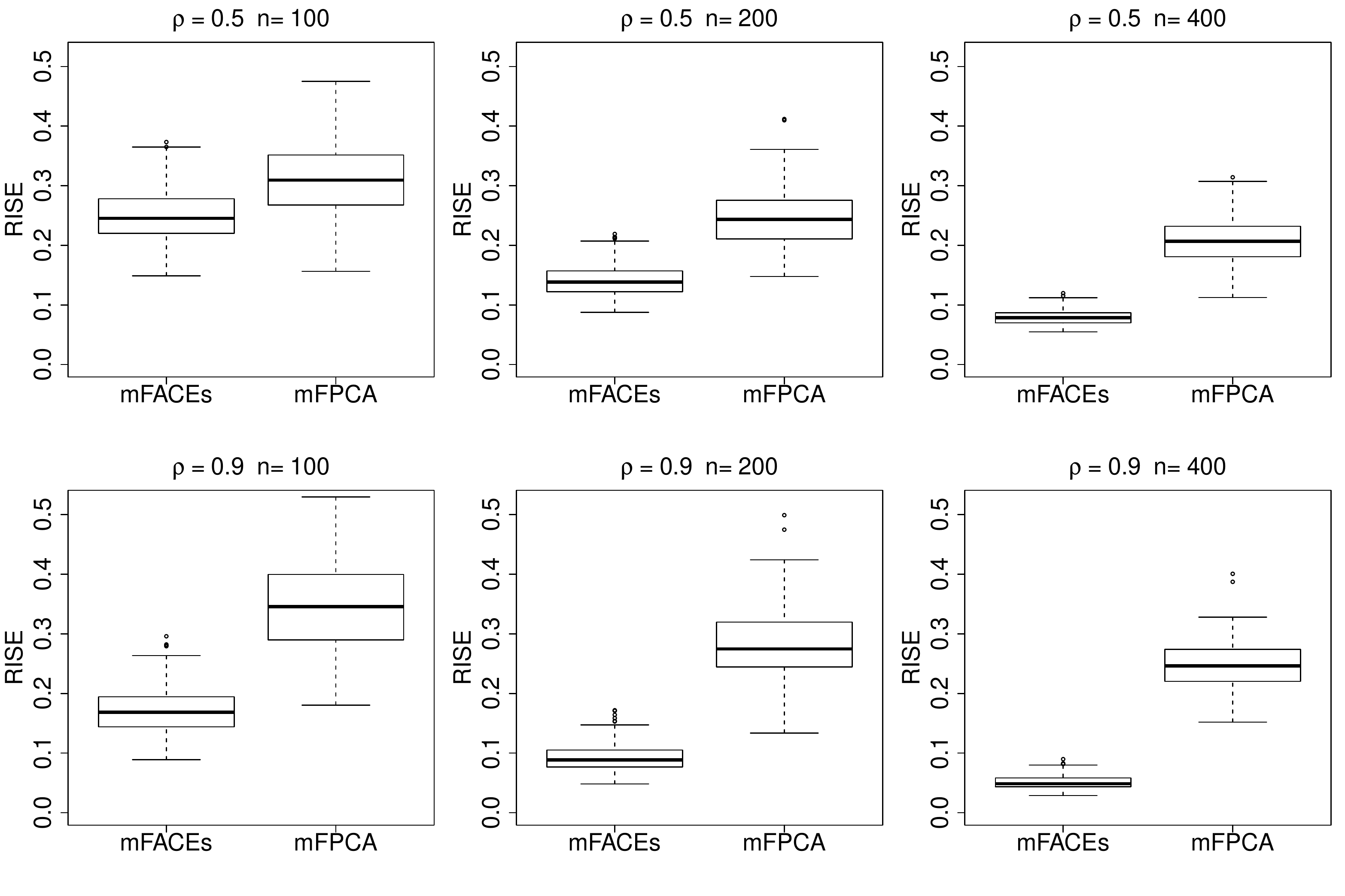}
	}
	\caption{\label{fig:cov_sim2}Boxplots of RISEs of mFACEs and mFPCA for estimating the covariance function.}
\end{figure}

Figures~\ref{fig:phi_sim2} and~\ref{fig:lam_sim2} give boxplots of ISEs and violin plots of mFACEs and mFPCA for estimating the top two eigenfunctions and eigenvalues,
respectively. The top two eigenvalues account for about $60\%$ of the total variation in the functional data
for $\rho = 0.5$ and it is $80\%$ for $\rho =0.9$.
Figure~\ref{fig:phi_sim2} shows that while the two methods are overall comparable for estimating the 1st eigenfunction, mFACEs 
has a much better accuracy for estimating the second eigenfunction  than mFPCA. 
The violin plots in Figure~\ref{fig:lam_sim2} show that mFACEs outperforms mFPCA substantially for estimating both eigenvalues under all model conditions. 
The mFPCA always underestimates the eigenvalues as the variation of scores from univariate FPCA is smaller than the true variation and hence leads to underestimates of eigenvalues.

\begin{figure}[htp]
	\centering
	\scalebox{0.35}{
		\includegraphics[]{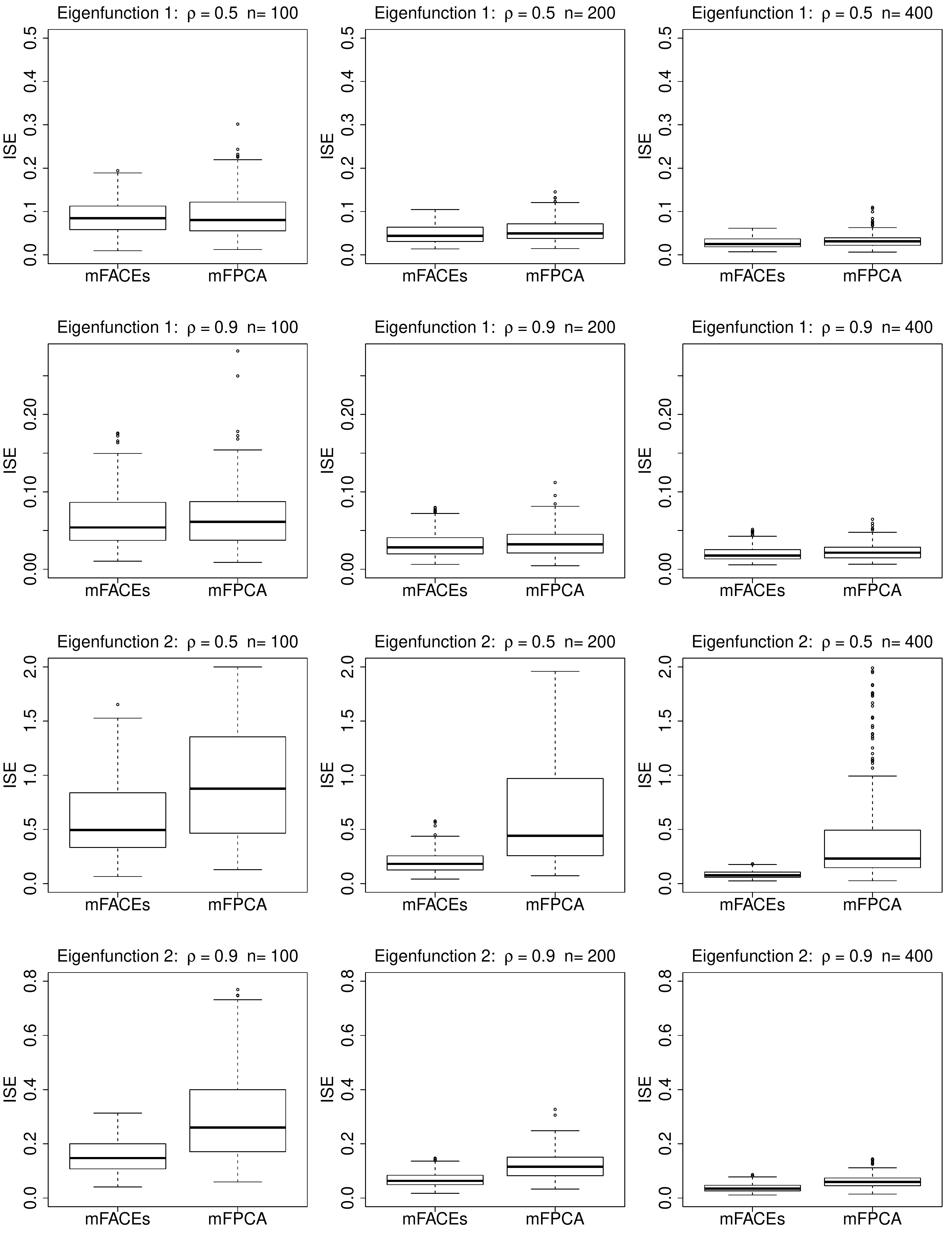}
	}
	\caption{\label{fig:phi_sim2} Boxplots of ISEs of mFACEs and mFPCA for estimating the top two eigenfunctions.}
\end{figure}

\begin{figure}[htp]
	\centering
	\scalebox{0.4}{
		\includegraphics[]{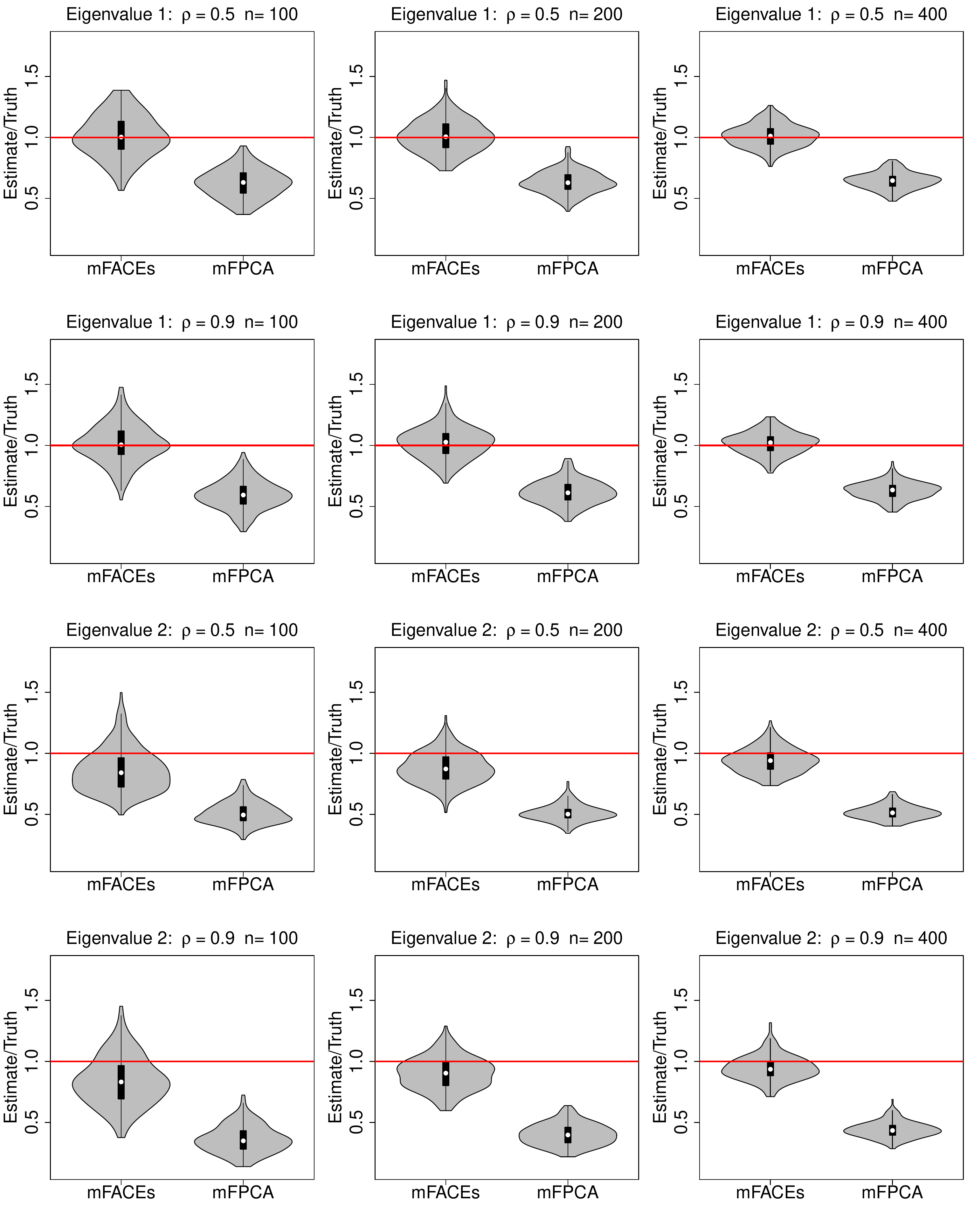}
	}
	\caption{\label{fig:lam_sim2} Violin plots of mFACEs and mFPCA for estimating the top two eigenvalues. The red horizontal lines indicate that the estimates are equal to the truth.}
\end{figure}

Finally, we consider the prediction of new subjects by mFACEs, mFPCA and mFPCA(CE).
We define the relative efficiencies
of different methods as
the ratios of MISEs  with respect to that of univariate FPCA; see  Figure~\ref{fig:pred_sim2}. 
Univariate FPCA is implemented in the \texttt{R} package \texttt{face} \citep{Xiao:16c}.
We have the following findings.
Under all model conditions, mFACEs has the smallest MISE, mFPCA(CE) has the second best performance, and mFPCA is close to univariate FPCA. 
Thus, on average mFACEs provides the most accurate curve prediction.
These results indicate that: 1) mFACEs has better covariance estimation than mFPCA(CE), and so is the prediction based on it; 
2) compared to mFPCA/univariate FPCA, mFPCA(CE) exploits the correlation information and hence results in better predictions.

\begin{figure}[htp]
	\centering
	\scalebox{0.4}{
		\includegraphics[]{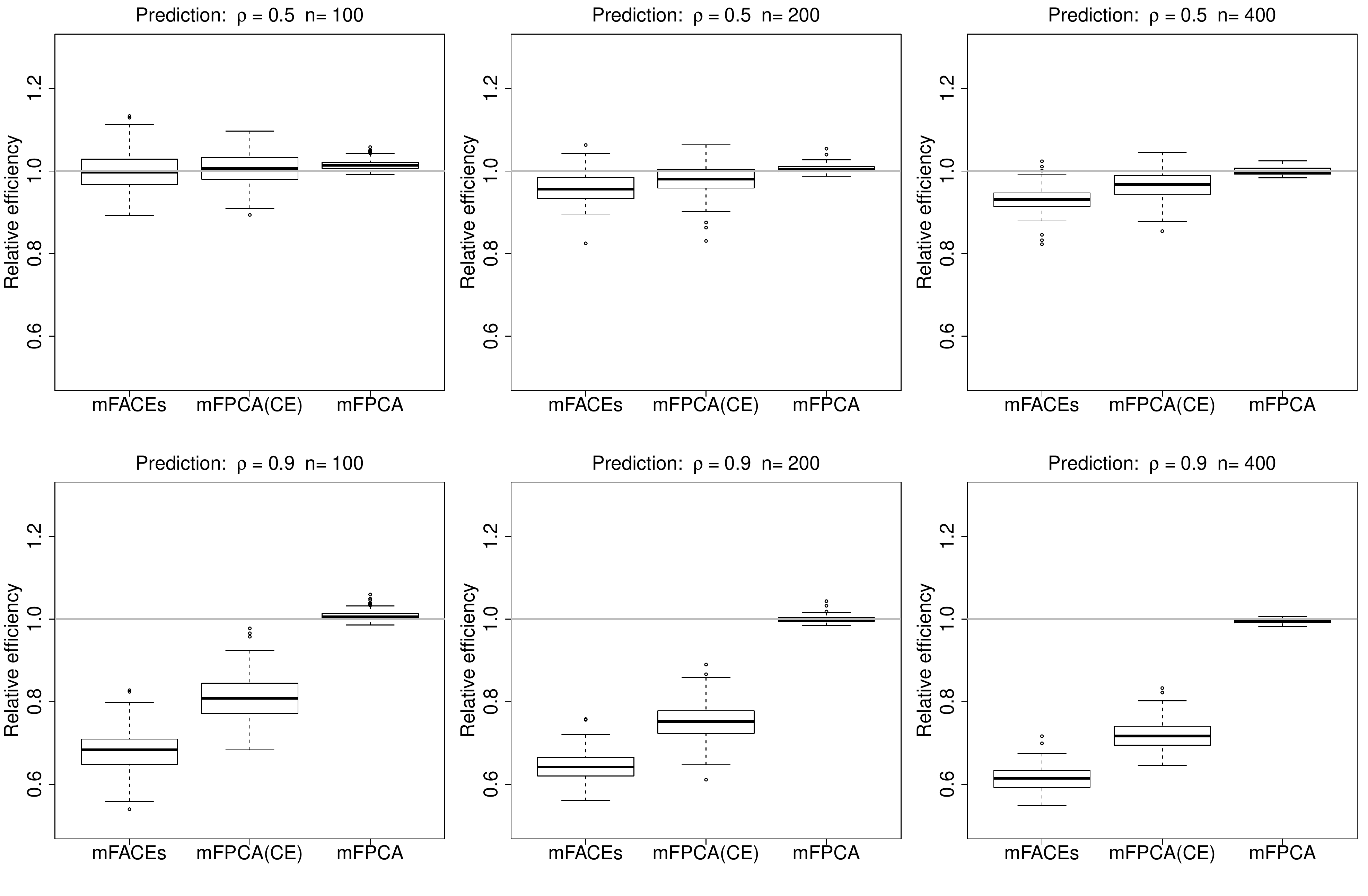}
	}
	\caption{\label{fig:pred_sim2}  Boxplots of relative efficiency of three methods for curve prediction. The gray horizontal lines indicate the MISEs for univariate FPCA.}
\end{figure}

In summary, mFACEs shows competing performance against alternative methods.
%
%
%
%
%
%

\section{Application to Alzheimer's Disease Study} \label{sec:application}
The Alzheimer's Disease Neuroimaging Initiative (ADNI) is a two-stage longitudinal observational study launched in year 2003 with the primary goal of investigating
whether serial neuroimags, biological markers, clinical and neuropsychological assessments can be combined to measure the progression of Alzheimer's disease (AD) \citep{weiner2017recent}.
The ADNI-1 data from the first stage contain $379$ patients with amnestic mild
cognitive impairment (MCI, a risk state for AD) at baseline who had at least one follow-up visit. 
Participants were assessed at baseline, $6$, $12$, $18$, $24$, and $36$ months with additional annual follow-ups included in the second stage of the study.
At each visit, various neuropsychological assessments, clinical measures, and brain images were collected.
The ADNI-2 data include  $424$ additional patients suffering from MCI and significant memory concern, with at least one
follow-up visit and longitudinal data collected over four years.
Thus, for the combined data, 
the total number of subjects is $803$, and the average number of visits is $4.72$.
The data are publicly available at \url{http://ida.loni.ucla.edu/}. 

We consider five longitudinal markers commonly measured in studies of AD with strong comparative predictive
value \citep{li2017prediction}.
Among the five markers,
Disease Assessment Scale-Cognitive 13 items (ADAS-Cog 13), Rey Auditory Verbal Learning Test immediate recall (RAVLT.imme), Rey Auditory Verbal Learning Test learning curve (RAVLT.learn), and Mini-Mental State Examination (MMSE) are neuropsychological assessments.
Functional Assessment Questionnaire (FAQ) is a functional and behavioral assessment.
High values of ADAS-Cog 13 and FAQ indicate a high-risk state for AD, whereas low values of RAVLT.imme, RAVLT.learn and MMSE reflect severe cognitive impairment. 
The longitudinal trajectories in ADNI-1 and ADNI-2 are defined on the same time domain with the largest follow-up time $96$ months from the start of ADNI-1 (time $0$).


\subsection{Multivariate FPCA via mFACEs}\label{sec:MFPCA}
We analyze the five longitudinal biomarkers using  mFACEs. 
For better visualization, we plot in Figure~\ref{fig:ADNI_cor}
the estimated correlation functions $\rho_{kk'}(s,t) = C_{kk'}(s,t)/\sqrt{C_{kk}(s,s) C_{k'k'}(t,t)}$.
The plot indicates of two groups of biomarkers: ADAS-Cog 13 and FAQ in one group whereas
RAVLT.imme, RAVLT.learn and MMSE in another group. 
The biomarkers within the groups are positively correlated and negatively correlated between groups,
which make sense as high values of ADAS-Cog 13 and FAQ and low values for the other biomarkers
suggest of AD.
Next, we display in Figure~\ref{fig:ADNI_eig} the two estimated (multivariate) eigenfunctions associated with
the top two estimated eigenvalues, which
account for $69\%$ and $11\%$ of the total variance in the functional part of the data, respectively.
The eigenfunctions reveal how the 5 biomarkers co-variate and how a subject's trajectories of biomarkers
deviate from the population mean. 
Indeed, we see from Figure~\ref{fig:ADNI_eig} that the first eigenfunction (solid curves) is below the zero-line for 
ADAS-Cog 13 and FAQ and above the zero-line for the other three biomarkers.
This means that the score corresponding to the first eigenfunction might be used as an indicator of AD.
Indeed,  a negative score for the first eigenfunction 
means  higher-than-population-mean values of the former while
lower-than-population-mean values of the latter, indicating more severe AD status.
The second eigenfunction (dashed curves) for the five biomarkers is below the zero line at first and then above it or the other way around,
potentially suggesting of a longitudinal pattern of the AD progression. Specifically, these subjects with a positive score for the second
eigenfunction will have higher ADAS-Cog 13/FAQ and lower RAVLT and MMSE over the months, suggesting of AD progression.
Finally, we illustrate  in Figure~\ref{fig:ADNI_curves} the predicted curves along with the associated $95\%$ point-wise confidence bands
for three subjects.
We focus on predicting the trajectories over the first four years as there are more observations.
We can see that the confidence bands are getting wider at the later time points because of fewer observations.

\begin{figure}[htp]
	\centering
	\scalebox{0.6}{
		\includegraphics[]{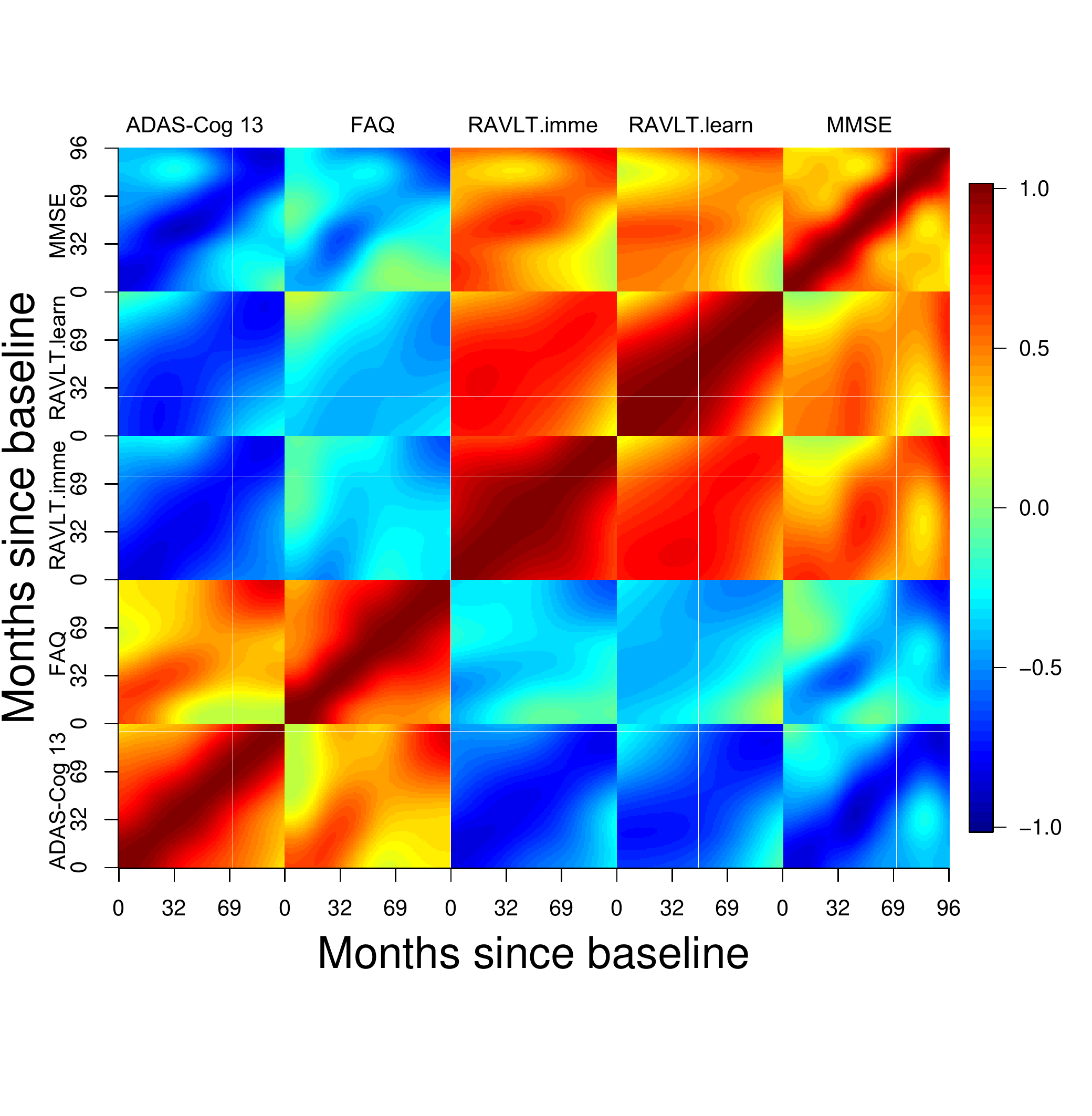}
	}
	\caption{\label{fig:ADNI_cor} Estimated correlation functions for the longitudinal markers.}
\end{figure}

\begin{figure}[htp]
	\centering
	\scalebox{0.3}{
		\includegraphics[]{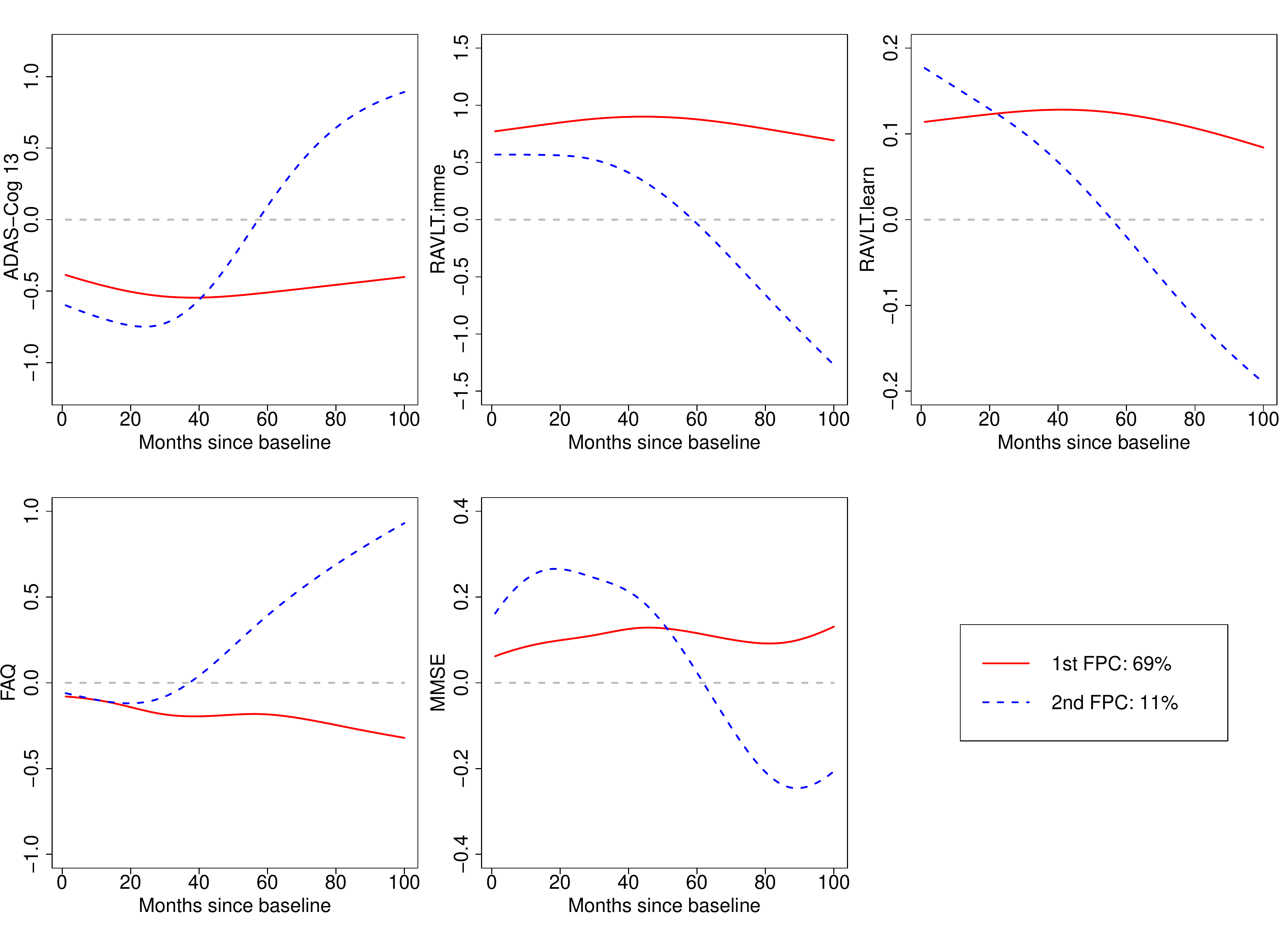}
	}
	\caption{\label{fig:ADNI_eig} Estimated top two eigenfunctions for the longitudinal markers.} 
\end{figure}

\begin{figure}[htp]
	\centering
	\scalebox{0.35}{
		\includegraphics[]{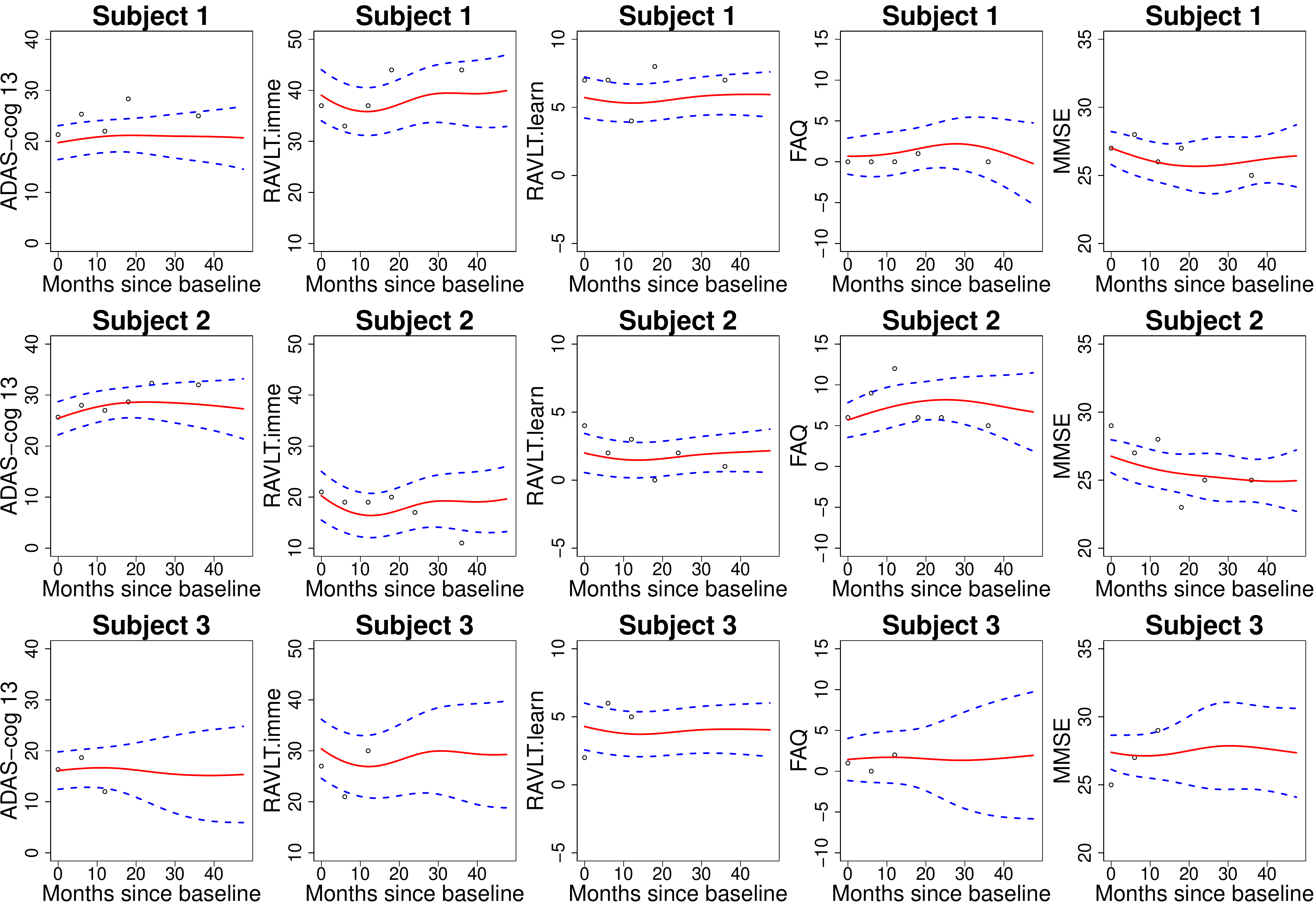}
	}
	\caption{\label{fig:ADNI_curves} Predicted subject-specific curves (red solid line) of the longitudinal markers and associated $95\%$ point-wise confidence bands (blue dashed line) for three subjects.}
\end{figure}

\subsection{Comparison of Prediction Performance of Different Methods}\label{sec:com_prediction}
We compare the proposed mFACEs with mFPCA and mFPCA(CE) for predicting the five longitudinal biomarkers. 
The prediction performance is evaluated by the average squared prediction errors (APE),
\begin{eqnarray*}
	\textnormal{APE}_k = \frac{1}{n}\sum_{i=1}^n\left[\frac{1}{m_{ik}}\sum_{j=1}^{m_{ik}} \left\{y_{ij}^{(k)} - \widehat{y}_{ij}^{(k)}\right\}^2  \right],
\end{eqnarray*}
where $\widehat{y}_{ij}^{(k)}$ is the predicted value of the $k$th biomarker for the $i$th subject at time $t_{ij}^{(k)}$.
We conduct two types of validation: an internal validation and an external validation.
For the internal validation, we perform a $10$-fold cross-validation to the combined data of ADNI-1 and ADNI-2.
For the external validation, we  fit the model using only the ADNI-1 data and then predict ADNI-2 data. 
Figure~\ref{fig:ADNI_valid} summarizes the results.
For simplicity, we  present the relative efficiency of APE,  which is the ratio of APEs of one method against the mFPCA. 
In both cases, mFACEs achieves better prediction accuracy than competing methods.  Note that mFPCA(CE) outperforms mFPCA  for predicting almost all biomarkers.
The results suggest that: 
1) mFACEs is better than competing methods for analyzing the longitudinal biomarkers.
2) exploiting the correlations between the biomarkers improve prediction.

\begin{figure}[htp]
	\centering
	\scalebox{0.3}{
		\includegraphics[]{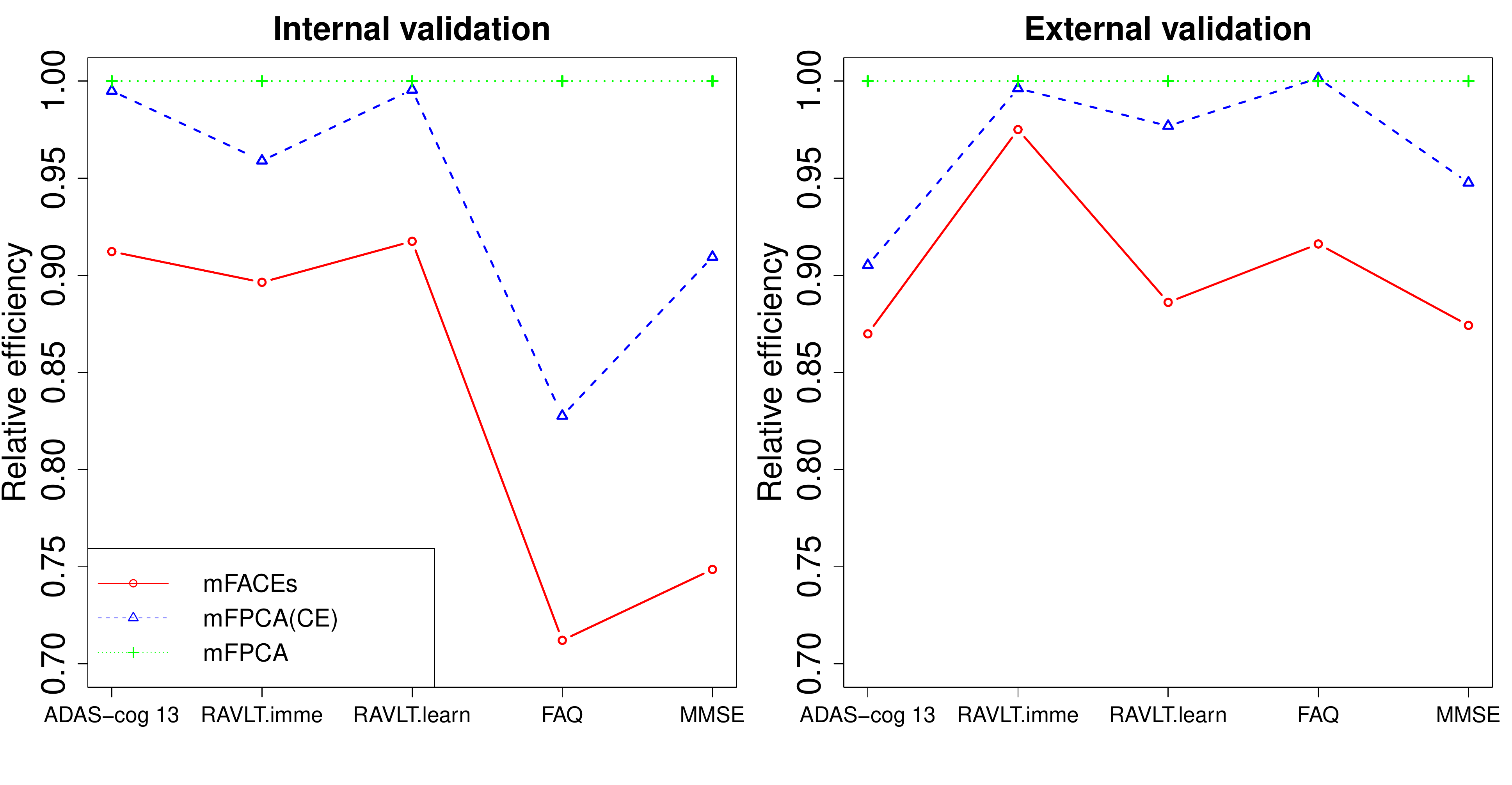}
	}
	\caption{\label{fig:ADNI_valid} The internal and external prediction validations for the ADNI longitudinal makers.} 
\end{figure}

\section{Discussion} \label{sec:discussion}
The prevalence of multivariate functional data has sparked much research interests in recent years. 
However, covariance estimation for multivariate sparse functional data remains underdeveloped.
We proposed a new method, mFACEs, and its features include:
1) a covariance smoothing framework is proposed to tackle multivariate sparse functional data;
2) an automatic and fast fitting algorithm is adopted to ensure the scalability of the method;
3) eigenfunctions and eigenvalues can be obtained through a one-time spectral decomposition, and eigenfunctions can be easily evaluated at any sampling points;
4) a multivariate extension of the conditional expectation approach \citep{Yao:05a} is derived to exploit correlations between outcomes.
The simulation study and the data example showed that mFACEs could better capture between-function correlations and thus give improved 
principal component analysis and curve prediction.


When the magnitude of functional data are quite different, one may first normalize the functional data, as recommended by \cite{Chiou:14}. 
One method of normalization is to rescale the functional data  using the estimated variance function $\widehat{C}_{kk}(t,t)^{-1/2}$ as in \cite{Chiou:14} and \cite{jacques2014model}. An alternative method is to use a global rescaling factor like $\left(\int \widehat{C}_{kk}(t,t) dt \right)^{-1/2}$ as in \cite{Happ:17}. Both methods can be easily incorporated into our proposed method. In our data analysis, we find that the results with normalization are very close to those without normalization, thus we present the results without normalization. 

Because multivariate FPCA is more complex than univariate FPCA, 
weak correlations between the functions and small sample size may offset the benefit of conducting multivariate FPCA, 
see Section $7.3$ in \cite{wong2017partially}. 
Thus, it is of future interest to develop practical tests to determine if correlations between multivariate functional data are different from 0.

The mFACEs method has been implemented in an R package \texttt{mfaces} and will be submitted to CRAN for public access.

\section*{Appendices}
\subsection*{Appendix A: Mean Function Estimation}
The smooth mean function $\mu^{(k)}(t)$ is approximated by the B-spline basis functions $f^{(k)}(t) = \sum_{1 \leq \gamma \leq c}\alpha_{\gamma}^{(k)}B_{\gamma}(t)$, where $\balpha_k = \left\{\alpha_1^{(k)},\ldots, \alpha_c^{(k)}\right\}^{\top}\in\real^c$ is a coefficient vector. 
For simplicity, we use the same set of B-spline bases as in the covariance function estimation.
We carry out univariate smoothing for each response using {\it P}-splines \citep{Eilers:96}
and $\balpha_k$ is obtained by minimizing
\begin{eqnarray}\label{eq:mean}
	\sum_{i=1}^n \sum_{j=1}^{m_{ik}} \left\{ f^{(k)}\left(t_{ij}^{(k)}\right) - y_{ij}^{(k)} \right\}^2 + \tau_{k} \|\bD\balpha_k\|^2,
\end{eqnarray} 
where $\tau_k$ is a nonnegative smoothing parameter to be selected by leave-one-subject-out cross validation
for the $k$th response.
Note that the penalty term is essentially equivalent to the integrated squared second derivative of $f^{(k)}$.
Denote the minimizer of~\eqref{eq:mean} by $\widehat{\balpha}_k$, then the estimate
of the mean function $\mu^{(k)}(t)$ is given by $\widehat{\mu}^{(k)}(t) = \sum_{1 \leq \gamma \leq c}\widehat{\alpha}_{\gamma}^{(k)}B_{\gamma}(t)$.

\subsection*{Appendix B: Proofs of Propositions~\ref{eq:eigen} and~\ref{eq:igcv}}
\begin{proof}[Proof of Proposition~\ref{eq:eigen}]
	Define  $\bbtilde(t) = \bG^{-\frac{1}{2}}\bb(t)$, then $\int \bbtilde(t)\bbtilde(t)^{\top} dt = \bI$.
	Define $\check{\bTheta}_{kk'}(t) = \bG^{\frac{1}{2}}\bTheta_{kk'}\bG^{\frac{1}{2}}$.
	According to \eqref{eq:cov_operator}, 
	\begin{eqnarray*}
		d_{\ell} \Psi_{\ell}^{(k)}(s) &=&
		\sum_{k^{\prime}=1}^p \int \bb(s)^{\top} \bTheta_{k k^{\prime}} \bb(t) \Psi_{\ell}^{(k^{\prime})}(t) dt \\
		&=&\bbtilde(s)^{\top} \sum_{k^{\prime}=1}^p   \check{\bTheta}_{k k^{\prime}} \int\bbtilde(t)\Psi_{\ell}^{(k')}(t) dt.
	\end{eqnarray*}
	Thus, $ \Psi_{\ell}^{(k)}(s) =  \bbtilde(s)^{\top} \bu_{\ell}^{(k)}$ with $\bu_{\ell}^{(k)} = d_{\ell}^{-1}\left\{\sum_{k^{\prime}=1}^p   \check{\bTheta}_{k k^{\prime}} \int\bbtilde(t)\Psi_{\ell}^{(k')}(t) dt\right\}\in\real^{c}$ and $\bu_{\ell} = (\bu_{\ell}^{(1),\top},\ldots, \bu_{\ell}^{(p),\top})^{\top}\in\real^{pc}$.
	Since $\sum_{k=1}^p\int \Psi_{\ell}^{(k)}(t)\Psi_{\ell'}^{(k)}(t) dt = 1_{\{\ell=\ell'\}}$, we derive that
	\begin{equation}
		\label{eq:ulk}
		\bu_{\ell}^{\top}\bu_{\ell'} = \sum_{k=1}^p \bu_{\ell}^{(k),\top}\bu_{\ell'}^{(k)} = 1_{\{\ell=\ell'\}}.
	\end{equation}
	By~\eqref{eq:mult_cov}, 
	$$C_{kk'}(s,t) = \bbtilde(s)^{\top}\check{\bTheta}_{kk'}\bbtilde(t) = \bbtilde(s)^{\top}\left\{\sum_{\ell\geq1}d_{\ell} \bu_{\ell}^{(k)}\bu_{\ell}^{(k'),\top}\right\}\bbtilde(t),$$
	which gives
	$$
	\check{\bTheta}_{kk'} = \sum_{\ell\geq1} d_{\ell}\bu_{\ell}^{(k)}\bu_{\ell}^{(k')}.
	$$
	As $1 \leq k, k^{\prime} \leq p$, the above is equivalent to
	\begin{eqnarray*}
		\underbrace{\begin{pmatrix} \check{\bTheta}_{11} & \ldots & \check{\bTheta}_{1p} \\ \vdots & \ddots & \vdots \\ \check{\bTheta}_{p1} & \ldots & \check{\bTheta}_{pp} \end{pmatrix}}_{:=\check{\bTheta}}
		&=& \sum_{\ell\geq 1} d_{\ell} \bu_{\ell}\bu_{\ell}^{\top}.
	\end{eqnarray*}
	Because of \eqref{eq:ulk},  $\bu_{\ell}$s are orthonormal eigenvectors of $\check{\bTheta}$ with  $d_{\ell}$s the corresponding eigenvalues.
	The proof is now complete.
\end{proof}

\begin{proof}[Proof of Proposition~\ref{eq:igcv}]
	
	
	By~\eqref{eq:igcv_a1},
	\begin{equation}\label{eq:igcv_a1_1}
		\iGCV  =\left\|\bChat\right\|^2 - 2\bbf^{\top}\bSigma^{-1}\bbf + \bbf^{\top}\bSigma^{-2}\bbf
		+ 2\sum_{i=1}^n\left(\bL_i\bSigma^{-1}\bbf-\bbf_i\right)^{\top}\bSigma^{-1}\left(\bL_i\bSigma^{-1}\bbf-\bbf_i\right).
	\end{equation}
	Since $\bSigma^{-1} = \bU\text{diag}(\bdtilde)\bU^{\top}$, we have
	\begin{equation}
		\label{eq:prop2:eq1}
		\bbf^{\top}\bSigma^{-1}\bbf = \bftilde^{\top}\text{diag}(\bdtilde)\bftilde = \bdtilde^{\top}(\bftilde\odot\bbftilde).
	\end{equation}
	Similarly,
	\begin{equation}
		\label{eq:prop2:eq2}
		\bbf^{\top}\bSigma^{-2}\bbf = \bftilde^{\top}\text{diag}(\bdtilde^2)\bftilde = (\bftilde\odot\bdtilde)^{\top}(\bftilde\odot\bdtilde).
	\end{equation}
	Next we derive that
	\begin{align*}
		&\left(\bL_i\bSigma^{-1}\bbf-\bbf_i\right)^{\top}\bSigma^{-1}\left(\bL_i\bSigma^{-1}\bbf-\bbf_i\right)\\
		=& \left(\bU^{\top}\bL_i\bSigma^{-1}\bbf-\bU^{\top}\bbf_i\right)^{\top}\text{diag}(\bdtilde)\left(\bU^{\top}\bL_i\bSigma^{-1}\bbf-\bU^{\top}\bbf_i\right)\\
		=&\left(\bLtilde_i\text{diag}(\bdtilde)\bftilde-\bftilde_i\right)^{\top}\text{diag}(\bdtilde)\left(\bLtilde_i\text{diag}(\bdtilde)\bftilde-\bftilde_i\right).
	\end{align*}
	It follows that
	\begin{equation}\label{eq:prop2:eq3}
		\begin{split}
			&\left(\bL_i\bSigma^{-1}\bbf-\bbf_i\right)^{\top}\bSigma^{-1}\left(\bL_i\bSigma^{-1}\bbf-\bbf_i\right)\\
			=&\bdtilde^{\top}\left[ \left\{\bLtilde_i(\bftilde\odot\bdtilde)\right\}\odot
			\left\{\bLtilde_i(\bftilde\odot\bdtilde)\right\}
			+ (\bftilde_i\odot \bftilde_i)\right] - 2\bdtilde^{\top} \left\{(\bftilde_i\bftilde^{\top})\odot\bLtilde_i\right\} \bdtilde.
		\end{split}
	\end{equation}
	Combining~\eqref{eq:igcv_a1_1}, \eqref{eq:prop2:eq1}, \eqref{eq:prop2:eq2} and \eqref{eq:prop2:eq3},
	the proof is complete.
\end{proof}


%
%
%

\subsection*{Appendix C: A Lemma}
\begin{lem}\label{lem1}
	The covariance operator with the covariance functions defined in Section 4.1 is positive semi-definite.
\end{lem}
\begin{proof}
	Let $\ba = (a_1,\ldots, a_p)^{\top}\in\real^p$
	and $\widetilde{X} = \ba^{\top}\bx$, then $\widetilde{X}$ is a stochastic process
	with covariance function $$\text{Cov}\left\{\widetilde{X}(s), \widetilde{X}(t)\right\}
	= \sum_{k k'} a_k a_{k'} C_{kk'}(s,t)
	=\sum_{k k'} a_k a_{k'} \left( \rho + (1-\rho) 1_{\{k=k'\}}\right)
	\bPhitilde_k(s)^{\top} \bPhitilde_{k'}(t),$$
	where $\bPhitilde_k(s) = \bPhi_k(s)^{\top}\bLambda_{kk}^{\frac{1}{2}}$.
	Let $\bPsi(s) = \sum_{k=1}^p a_k \bPhitilde_k(s)$.
	Then
	$$
	\text{Cov}\left\{\widetilde{X}(s), \widetilde{X}(t)\right\}
	= \rho \bPsi(s)^{\top}\bPsi(t)
	+ (1-\rho) \sum_{k} a_k^2 \bPhitilde_k(s)^{\top}\bPhitilde_{k}(t).
	$$
	which is always positive semi-definite and the proof is complete.
\end{proof}

%
%
%
%
%

\newpage
\bibliography{ref}

\begin{thebibliography}{}

\bibitem[\protect\citeauthoryear{Berrendero, Justel, and Svarc}{Berrendero
  et~al.}{2011}]{Berrendero:11}
Berrendero, J., A.~Justel, and M.~Svarc (2011).
\newblock Principal components for multivariate functional data.
\newblock {\em Computational Statistics \& Data Analysis\/}~{\em 55\/}(9),
  2619--2634.

\bibitem[\protect\citeauthoryear{Cai and Yuan}{Cai and
  Yuan}{2010}]{cai2010nonparametric}
Cai, T. and M.~Yuan (2010).
\newblock Nonparametric covariance function estimation for functional and
  longitudinal data.
\newblock {\em University of Pennsylvania and Georgia inistitute of
  technology\/}.

\bibitem[\protect\citeauthoryear{Chiou, Chen, and Yang}{Chiou
  et~al.}{2014}]{Chiou:14}
Chiou, J.-M., Y.-T. Chen, and Y.-F. Yang (2014).
\newblock Multivariate functional principal component analysis: A normalization
  approach.
\newblock {\em Statistica Sinica\/}, 1571--1596.

\bibitem[\protect\citeauthoryear{Chiou and M{\"u}ller}{Chiou and
  M{\"u}ller}{2014}]{chiou2014}
Chiou, J.-M. and H.-G. M{\"u}ller (2014).
\newblock Linear manifold modelling of multivariate functional data.
\newblock {\em Journal of the Royal Statistical Society: Series B (Statistical
  Methodology)\/}~{\em 76\/}(3), 605--626.

\bibitem[\protect\citeauthoryear{Chiou and M{\"u}ller}{Chiou and
  M{\"u}ller}{2016}]{chiou2016pairwise}
Chiou, J.-M. and H.-G. M{\"u}ller (2016).
\newblock A pairwise interaction model for multivariate functional and
  longitudinal data.
\newblock {\em Biometrika\/}~{\em 103\/}(2), 377--396.

\bibitem[\protect\citeauthoryear{Eilers and Marx}{Eilers and
  Marx}{1996}]{Eilers:96}
Eilers, P. and B.~Marx (1996).
\newblock {Flexible smoothing with B-splines and penalties (with Discussion)}.
\newblock {\em Statist. Sci.\/}~{\em 11}, 89--121.

\bibitem[\protect\citeauthoryear{Eilers and Marx}{Eilers and
  Marx}{2003}]{Eilers:03}
Eilers, P. and B.~Marx (2003).
\newblock {Multivariate calibration with temperature interaction using
  two-dimensional penalized signal regression}.
\newblock {\em Chemometrics and Intelligent Laboratory Systems\/}~{\em 66},
  159--174.

\bibitem[\protect\citeauthoryear{Goldsmith, Crainiceanu, Caffo, and
  Reich}{Goldsmith et~al.}{2012}]{Goldsmith:12}
Goldsmith, J., C.~Crainiceanu, B.~Caffo, and D.~Reich (2012).
\newblock Longitudinal penalized functional regression for cognitive outcomes
  on neuronal tract measurements.
\newblock {\em Journal of the Royal Statistical Society: Series C (Applied
  Statistics)\/}~{\em 61}, 453--469.

\bibitem[\protect\citeauthoryear{Greven, Crainiceanu, Caffo, and Reich}{Greven
  et~al.}{2010}]{Greven:10}
Greven, S., C.~Crainiceanu, B.~Caffo, and D.~Reich (2010).
\newblock {Longitudinal functional principal component}.
\newblock {\em Electronic J. Statist.\/}~{\em 4}, 1022--1054.

\bibitem[\protect\citeauthoryear{Happ and Greven}{Happ and
  Greven}{2018}]{Happ:17}
Happ, C. and S.~Greven (2018).
\newblock Multivariate functional principal component analysis for data
  observed on different (dimensional) domains.
\newblock {\em Journal of the American Statistical Association\/}~{\em
  113\/}(522), 649--659.

\bibitem[\protect\citeauthoryear{Huang, Li, and Guan}{Huang
  et~al.}{2014}]{huang2014joint}
Huang, H., Y.~Li, and Y.~Guan (2014).
\newblock Joint modeling and clustering paired generalized longitudinal
  trajectories with application to cocaine abuse treatment data.
\newblock {\em Journal of the American Statistical Association\/}~{\em
  109\/}(508), 1412--1424.

\bibitem[\protect\citeauthoryear{Jacques and Preda}{Jacques and
  Preda}{2014}]{jacques2014model}
Jacques, J. and C.~Preda (2014).
\newblock Model-based clustering for multivariate functional data.
\newblock {\em Computational Statistics \& Data Analysis\/}~{\em 71}, 92--106.

\bibitem[\protect\citeauthoryear{James, Hastie, and Sugar}{James
  et~al.}{2000}]{James:00}
James, G., T.~Hastie, and C.~Sugar (2000).
\newblock Principal component models for sparse functional data.
\newblock {\em Biometrika\/}~{\em 87}, 587--602.

\bibitem[\protect\citeauthoryear{Kowal, Matteson, and Ruppert}{Kowal
  et~al.}{2017}]{kowal2017bayesian}
Kowal, D.~R., D.~S. Matteson, and D.~Ruppert (2017).
\newblock A bayesian multivariate functional dynamic linear model.
\newblock {\em Journal of the American Statistical Association\/}~{\em
  112\/}(518), 733--744.

\bibitem[\protect\citeauthoryear{Leng and M\"uller}{Leng and
  M\"uller}{2006}]{Leng:06}
Leng, X. and H.~M\"uller (2006).
\newblock Classification using functional data analysis for temporal gene
  expression data.
\newblock {\em Bioinformatics\/}~{\em 22}, 68--76.

\bibitem[\protect\citeauthoryear{Li, Huang, Zhu, and Initiative}{Li
  et~al.}{2017}]{li2017functional}
Li, J., C.~Huang, H.~Zhu, and A.~D.~N. Initiative (2017).
\newblock A functional varying-coefficient single-index model for functional
  response data.
\newblock {\em Journal of the American Statistical Association\/}~{\em
  112\/}(519), 1169--1181.

\bibitem[\protect\citeauthoryear{Li, Chan, Doody, Quinn, and Luo}{Li
  et~al.}{2017}]{li2017prediction}
Li, K., W.~Chan, R.~S. Doody, J.~Quinn, and S.~Luo (2017).
\newblock Prediction of conversion to alzheimer¡¯s disease with longitudinal
  measures and time-to-event data.
\newblock {\em Journal of Alzheimer's Disease\/}~{\em 58\/}(2), 361--371.

\bibitem[\protect\citeauthoryear{Li, Wang, and Carroll}{Li
  et~al.}{2013}]{li2013selecting}
Li, Y., N.~Wang, and R.~J. Carroll (2013).
\newblock Selecting the number of principal components in functional data.
\newblock {\em Journal of the American Statistical Association\/}~{\em
  108\/}(504), 1284--1294.

\bibitem[\protect\citeauthoryear{Lindquist}{Lindquist}{2012}]{Lindquist:12}
Lindquist, M. (2012).
\newblock Functional causal mediation analysis with an application to brain
  connectivity.
\newblock {\em Journal of the American Statistical Association\/}~{\em
  107\/}(500), 1297--1309.

\bibitem[\protect\citeauthoryear{Luo and Qi}{Luo and
  Qi}{2017}]{luo2017function}
Luo, R. and X.~Qi (2017).
\newblock Function-on-function linear regression by signal compression.
\newblock {\em Journal of the American Statistical Association\/}~{\em
  112\/}(518), 690--705.

\bibitem[\protect\citeauthoryear{Morris, Arroyo, Coull, Ryan, Herrick, and
  Gortmaker}{Morris et~al.}{2006}]{Morris:06}
Morris, J., C.~Arroyo, B.~Coull, L.~Ryan, R.~Herrick, and S.~Gortmaker (2006).
\newblock Using wavelet-based functional mixed models to characterize
  population heterogeneity in accelerometer profiles: A case study.
\newblock {\em Journal of the American Statistical Association\/}~{\em
  101\/}(476), 1352--1364.

\bibitem[\protect\citeauthoryear{Park and Ahn}{Park and
  Ahn}{2017}]{park2017clustering}
Park, J. and J.~Ahn (2017).
\newblock Clustering multivariate functional data with phase variation.
\newblock {\em Biometrics\/}~{\em 73\/}(1), 324--333.

\bibitem[\protect\citeauthoryear{Peng and Paul}{Peng and Paul}{2009}]{Peng:09}
Peng, J. and D.~Paul (2009).
\newblock A geometric approach to maximum likelihood estimation of functional
  principal components from sparse longitudinal data.
\newblock {\em J. Comput. Graph. Stat.\/}~{\em 18}, 995--1015.

\bibitem[\protect\citeauthoryear{Petersen and M{\"u}ller}{Petersen and
  M{\"u}ller}{2016}]{petersen2016frechet}
Petersen, A. and H.-G. M{\"u}ller (2016).
\newblock Fr{\'e}chet integration and adaptive metric selection for
  interpretable covariances of multivariate functional data.
\newblock {\em Biometrika\/}~{\em 103\/}(1), 103--120.

\bibitem[\protect\citeauthoryear{Qi and Luo}{Qi and Luo}{2018}]{qi2018function}
Qi, X. and R.~Luo (2018).
\newblock Function-on-function regression with thousands of predictive curves.
\newblock {\em Journal of Multivariate Analysis\/}~{\em 163}, 51--66.

\bibitem[\protect\citeauthoryear{Qiao, Guo, and James}{Qiao
  et~al.}{2019}]{qiao2017functional}
Qiao, X., S.~Guo, and G.~M. James (2019).
\newblock Functional graphical models.
\newblock {\em Journal of the American Statistical Association\/}~{\em
  114\/}(525), 211--222.

\bibitem[\protect\citeauthoryear{Ramsay and Silverman}{Ramsay and
  Silverman}{2005}]{Ramsay:05}
Ramsay, J. and B.~Silverman (2005).
\newblock {\em Functional data analysis}.
\newblock New York: Springer.

\bibitem[\protect\citeauthoryear{Reimherr and Nicolae}{Reimherr and
  Nicolae}{2014}]{Reimherr:14}
Reimherr, M. and D.~Nicolae (2014).
\newblock A functional data analysis approach for genetic association studies.
\newblock {\em The Annals of Applied Statistics\/}~{\em 8}, 406--429.

\bibitem[\protect\citeauthoryear{Reimherr and Nicolae}{Reimherr and
  Nicolae}{2016}]{Reimherr:16}
Reimherr, M. and D.~Nicolae (2016).
\newblock Estimating variance components in functional linear models with
  applications to genetic heritability.
\newblock {\em Journal of the American Statistical Association\/}~{\em 111},
  407--422.

\bibitem[\protect\citeauthoryear{Reiss and Ogden}{Reiss and
  Ogden}{2010}]{Reiss:10}
Reiss, P. and R.~Ogden (2010).
\newblock Functional generalized linear models with images as predictors.
\newblock {\em Biometrics\/}~{\em 66}, 61--69.

\bibitem[\protect\citeauthoryear{Saporta}{Saporta}{1981}]{Saporta:81}
Saporta, G. (1981).
\newblock M{\'e}thodes exploratoires d¡¯analyse de donn{\'e}es temporelles.
\newblock {\em Cahiers du bureau universitaire de recherche
  op{\'e}rationnelle\/}~(37-38).

\bibitem[\protect\citeauthoryear{Seber}{Seber}{2007}]{Seber:07}
Seber, G. (2007).
\newblock {\em A Matrix Handbook for Statisticians}.
\newblock New Jersey: Wiley-Inter\-sci\-ence.

\bibitem[\protect\citeauthoryear{Weiner, Veitch, Aisen, Beckett, Cairns, Green,
  Harvey, Jack, Jagust, Morris, et~al.}{Weiner et~al.}{2017}]{weiner2017recent}
Weiner, M.~W., D.~P. Veitch, P.~S. Aisen, L.~A. Beckett, N.~J. Cairns, R.~C.
  Green, D.~Harvey, C.~R. Jack, W.~Jagust, J.~C. Morris, et~al. (2017).
\newblock Recent publications from the alzheimer's disease neuroimaging
  initiative: Reviewing progress toward improved ad clinical trials.
\newblock {\em Alzheimer's \& dementia: the journal of the Alzheimer's
  Association\/}~{\em 13\/}(4), e1--e85.

\bibitem[\protect\citeauthoryear{Wong, Li, and Zhu}{Wong
  et~al.}{2019}]{wong2017partially}
Wong, R.~K., Y.~Li, and Z.~Zhu (2019).
\newblock Partially linear functional additive models for multivariate
  functional data.
\newblock {\em Journal of the American Statistical Association\/}~{\em
  114\/}(525), 406--418.

\bibitem[\protect\citeauthoryear{Wong and Zhang}{Wong and
  Zhang}{2019}]{Wong2019}
Wong, R.~K. and X.~Zhang (2019).
\newblock Nonparametric operator-regularized covariance function estimation for
  functional data.
\newblock {\em Computational statistics \& data analysis\/}~{\em 131}, 131 --
  144.

\bibitem[\protect\citeauthoryear{Wood}{Wood}{2000}]{Wood:00}
Wood, S.~N. (2000).
\newblock Modelling and smoothing parameter estimation with multiple quadratic
  penalties.
\newblock {\em Journal of the Royal Statistical Society: Series B (Statistical
  Methodology)\/}~{\em 62\/}(2), 413--428.

\bibitem[\protect\citeauthoryear{Xiao, Huang, Schrack, Ferrucci, Zipunnikov,
  and Crainiceanu}{Xiao et~al.}{2015}]{Xiao:15}
Xiao, L., L.~Huang, J.~Schrack, L.~Ferrucci, V.~Zipunnikov, and C.~Crainiceanu
  (2015).
\newblock {Quantifying the life-time circadian rhythm of physical activity: a
  covariate-dependent functional approach}.
\newblock {\em Biostatistics\/}~{\em 16}, 352--367.

\bibitem[\protect\citeauthoryear{Xiao, Li, Checkley, and Crainiceanu}{Xiao
  et~al.}{2017}]{Xiao:16c}
Xiao, L., C.~Li, W.~Checkley, and C.~Crainiceanu (2017).
\newblock R package {\it face}: Fast covariance estimation for sparse
  functional data (version 0.1-4).
\newblock URL:\url{http://cran.r-project.org/web/packages/face/index.html}.

\bibitem[\protect\citeauthoryear{Xiao, Li, Checkley, and Crainiceanu}{Xiao
  et~al.}{2018}]{Xiao:16b}
Xiao, L., C.~Li, W.~Checkley, and C.~Crainiceanu (2018).
\newblock Fast covariance estimation for sparse functional data.
\newblock {\em Statistics and Computing\/}~{\em 28}, 511--522.

\bibitem[\protect\citeauthoryear{Xu and Huang}{Xu and Huang}{2012}]{Xu:12}
Xu, G. and J.~Huang (2012).
\newblock {Asymptotic optimality and efficient computation of the
  leave-subject-out cross-validation}.
\newblock {\em Ann. Statist.\/}~{\em 40}, 3003--3030.

\bibitem[\protect\citeauthoryear{Yao, M\"uller, and Wang}{Yao
  et~al.}{2005}]{Yao:05a}
Yao, F., H.~M\"uller, and J.~Wang (2005).
\newblock {Functional data analysis for sparse longitudinal data}.
\newblock {\em J. Amer. Statist. Assoc.\/}~{\em 100}, 577--590.

\bibitem[\protect\citeauthoryear{Zhou, Huang, and Carroll}{Zhou
  et~al.}{2008}]{zhou2008joint}
Zhou, L., J.~Z. Huang, and R.~J. Carroll (2008).
\newblock Joint modelling of paired sparse functional data using principal
  components.
\newblock {\em Biometrika\/}~{\em 95\/}(3), 601--619.

\bibitem[\protect\citeauthoryear{Zhou, Shen, Wolfe, et~al.}{Zhou
  et~al.}{1998}]{zhou1998local}
Zhou, S., X.~Shen, D.~Wolfe, et~al. (1998).
\newblock Local asymptotics for regression splines and confidence regions.
\newblock {\em Annals of Statistics\/}~{\em 26\/}(5), 1760--1782.

\bibitem[\protect\citeauthoryear{Zhu, Brown, and Morris}{Zhu
  et~al.}{2012}]{Zhu:12}
Zhu, H., P.~Brown, and J.~Morris (2012).
\newblock Robust classification of functional and quantitative image data using
  functional mixed models.
\newblock {\em Biometrics\/}~{\em 68}, 1260--1268.

\bibitem[\protect\citeauthoryear{Zhu, Li, and Kong}{Zhu
  et~al.}{2012}]{zhu2012multivariate}
Zhu, H., R.~Li, and L.~Kong (2012).
\newblock Multivariate varying coefficient model for functional responses.
\newblock {\em Annals of statistics\/}~{\em 40\/}(5), 2634.

\bibitem[\protect\citeauthoryear{Zhu, Morris, Wei, and Cox}{Zhu
  et~al.}{2017}]{zhu2017multivariate}
Zhu, H., J.~S. Morris, F.~Wei, and D.~D. Cox (2017).
\newblock Multivariate functional response regression, with application to
  fluorescence spectroscopy in a cervical pre-cancer study.
\newblock {\em Computational statistics \& data analysis\/}~{\em 111}, 88--101.

\bibitem[\protect\citeauthoryear{Zhu, Strawn, and Dunson}{Zhu
  et~al.}{2016}]{zhu2016bayesian}
Zhu, H., N.~Strawn, and D.~B. Dunson (2016).
\newblock Bayesian graphical models for multivariate functional data.
\newblock {\em Journal of Machine Learning Research\/}~{\em 17\/}(204), 1--27.

\end{thebibliography}
\bibliographystyle{chicago}

\end{document}